\documentclass[a4paper, notitlepage, 12pt]{article}
\usepackage{geometry}
\fontfamily{times}
\usepackage{setspace,relsize}               
\usepackage{moreverb}                        
\usepackage{xurl}
\usepackage{hyperref}
\hypersetup{colorlinks=true,citecolor=blue}
\usepackage{amsmath}
\usepackage{mathtools} 
\usepackage{amsthm}
\usepackage{amssymb}
\usepackage{indentfirst}
\usepackage{todonotes}
\usepackage[authoryear,round]{natbib}
\usepackage[pdftex]{lscape}
\usepackage[utf8]{inputenc}
\usepackage{booktabs}
\usepackage{multirow}
\usepackage[T1]{fontenc}
\usepackage{algorithm}
\usepackage{algpseudocode}

\usepackage{libertine}
\usepackage[libertine]{newtxmath}

\title{\vspace{-9ex}\centering \bf Adaptive truncation of infinite sums: applications to Statistics}
\author{
Luiz Max Carvalho \& Wellington J. Silva\\
School of Applied Mathematics, Getulio Vargas Foundation, Brazil.\\
\&\\
Guido A. Moreira\\
Bavarian-Nordic, Germany.
}
\DeclareRobustCommand{\stirling}{\genfrac\{\}{0pt}{}}

\newtheorem{theorem}{Theorem}[]

\newtheorem{lemma}{Lemma}[]
\newtheorem{proposition}{Proposition}[]
\newtheorem{remark}{Remark}[]
\newtheorem{definition}{Definition}[]
\begin{document}
\maketitle

\begin{abstract}
In statistics, it is often necessary to compute sums of infinite series, especially when marginalising over discrete latent variables.
This has become increasingly relevant with the popularization of gradient-based techniques (e.g., Hamiltonian Monte Carlo) for Bayesian inference, for which discrete latent variables are challenging to handle.
For well-known infinite series, like the Hurwitz Zeta function or the Conway-Maxwell Poisson normalising constant, custom algorithms have been developed to exploit specific features of each problem.
However, general techniques that apply to a wide range of problems with limited input from the user are less established.
Here we employ basic results from the theory of infinite series to investigate general, problem-agnostic algorithms to approximate (truncate) infinite sums within an arbitrary tolerance $\varepsilon > 0$ and provide robust computational implementations with provable guarantees.
We compare two tentative solutions to estimating the infinite sum of interest: (i) a ``naive'' approach that sums terms until the terms are below the threshold $\varepsilon > 0$;  (ii) a ``bounding pair'' strategy based on trapping the true value between two partial sums.
We demonstrate the regularity conditions under which each method guarantees that the truncated sum is within the desired tolerance.
We show that while this first approach is widely used, it often fails to meet all necessary conditions, limiting its applicability.
In contrast, the second method, although more mathematically restrictive, provides guarantees for a broad class of series commonly encountered in statistics, in particular, for all the problems discussed here.
A comparison of computing times is also provided, along with a detailed discussion of numerical issues in practical implementations.
We discuss various statistical applications, including raw and factorial moments and count models with observation error.
Finally, detailed illustrations in the form of noisy MCMC for Bayesian inference and maximum marginal likelihood estimation are presented.

\end{abstract}

\newpage
\tableofcontents
\newpage
\section{Introduction}
\label{sec:intro}

Infinite series find use in numerous statistical applications, including estimating normalising constants in doubly-intractable problems~\citep{Wei2017,Gaunt2019}, evaluating the density of Tweedie distributions~\citep{Dunn2005} and Bayesian nonparametrics~\citep{Griffin2016,Burns2023}.
More generally, various applications of Rao-Blackwellisation~\citep{Robert2021} in Markov Chain Monte Carlo depend on marginalising over discrete latent variables, although the literature on this topic is significantly more scant on this topic -- but see \cite{Navarro2009}.
With the advent of powerful gradient-based algorithms such as the Metropolis-adjusted Langevin algorithm (MALA, \cite{Roberts2002}) and dynamic Hamiltonian Monte Carlo (dHMC, \cite{Betancourt2017,Carpenter2017}), the need for marginalising out discrete variables was made even more apparent, as these variables lack the differential structure needed in order to make full use of these algorithms.
Another relevant class of problems is maximising the likelihood function when the model in question uses a discrete distribution whose normalising constant is not known analytically.

While it is possible to find solutions and respective computational implementations that offer \textit{ad hoc} numerical solutions such as Wolfram Alpha (\url{www.wolframalpha.com}),  many statistical problems such as marginalisation and maximum likelihood need to deal with parameters varying at each iteration of the respective algorithm (optimisation or MCMC).
Another common situation is having specific approximations for particular situations, such as the normalising constant of the Conway-Maxwell Poisson distribution~\cite{Gaunt2019} or the Riemann and Hurwitz zeta functions~\cite{Ferreira2004}.
While some of these approximations come with theoretical guarantees and are in general quite efficient -- in the sense of performing fewer computations -- they are also restricted to a narrow class of problems.
A direct consequence is that each researcher is left to their own devices when it comes to providing a solution to their own infinite summation problem.
Critically, calculations should provide both precision and time efficiency.
In this paper, we endeavour to provide a general-purpose method that can replace some commonly used approaches without guarantees.
Additionally, this manuscript can be used as a benchmark to compare the results and computation time of methods tailored to specific problems.
For convenience, we also provide implementations in an R \citep{R} \textit{via} the \textbf{sumR} package (\url{https://cran.r-project.org/web/packages/sumR/index.html}) and in Python in the package \textbf{InfSumPy} (\url{https://pypi.org/project/InfSumPy}).

We discuss adaptive algorithms for approximating infinite sums that are applicable to a plethora of real-world statistical situations, such as computing normalising constants, raw and factorial moments and marginalisation in count models with observation error.
We give theoretical guarantees for the analytical quality (error) of approximation and also discuss technical issues involved in implementing a numerically-stable algorithm.
The remainder of this paper is organised as follows: after some preliminary results are reviewed in Section~\ref{sec:prelim}, Section~\ref{sec:approaches} discusses two adaptive approaches to computing the infinite sum approximately and addresses their merits and pitfalls, as well as their theoretical guarantees.
An exposition of the technical issues involved in numerically stable implementations is provided in Section~\ref{sec:technical}.
We discuss common classes of statistical examples in Section~\ref{sec:applications} and empirical illustrations for noisy Markov chain Monte Carlo and marginal maximum likelihood estimation are carried out in Sections~\ref{sec:comp_noisy_mcmc} and~\ref{sec:mmle_erlang}, respectively.
We finish with a discussion and an overview of avenues for future research in Section~\ref{sec:discussion}.

\subsection{Preliminaries}
\label{sec:prelim}

In this section a few basic concepts and results from the theory of infinite series are reviewed.
They will be needed in the remainder of the paper as the provable guarantees for the proposed approximation schemes rely on them.
The interested reader can find a good resource in~\cite{Rudin1964}.

Many infinite summation problems do not admit closed-form solutions, and one is left with the problem of truncating a finite sum to achieve sufficient accuracy in computational applications.
Let $\left(a_n\right)_{n\geq 0}$ be a non-negative, absolutely convergent series (see Definition~\ref{def:absolutely_convergent} below).
In the remainder of this paper, it will be convenient to define $S_K := \sum_{n=0}^K a_n$.
Now, suppose we want to approximate the quantity
\begin{equation}
 \label{eq:the_sum}
 S := \sum_{n=0}^\infty a_n,
\end{equation}
with an error of at most $\varepsilon >0$, i.e., we want to obtain $\hat{S}$ such that $|\hat{S}-S| \leq \varepsilon$.
In Statistics, problems usually take the form of $a_n = p(n)f(n)$, where $p$ is a (potentially unnormalised) probability mass function (p.m.f.) and $f$ is a measurable function with respect to the probability measure associated with $p$. 
This framework is sufficiently broad to accommodate a range of statistical problems.
For instance, when $p$ is not normalised and $f(n) = 1$ for all $n$,  computing $S$ amounts to computing a normalising constant, whereas when $p$ is normalised and $f$ is the identity function, one is then concerned with computing the expected value.

Given the broad range of (statistical) applications where the problem of truncating infinite series arises, it is perhaps no surprise that no unified framework appears to exist.
A common tactic is to arbitrate a large integer $K$ and use $S_K$ as the estimate for $S$~\citep{Royle2004,Aleshin2021,Benson2021}.
For many applications it is hard to compute the truncation bound explicitly in order to guarantee the truncation is within a tolerance $\varepsilon$ -- see~\cite{Navarro2009} for an example where such bounds can be obtained explicitly.
This `fixed upper bound' approach thus usually comes with no truncation error guarantees.
Further, in some situations $S_K$ might give approximations that have smaller errors than $\varepsilon$ but require increased computation time.
It is therefore desirable to investigate adaptive truncation algorithms where $K$ can be chosen automatically or semi-automatically, so as to provide simultaneously reliable and potentially less onerous approximations.

The assumption that the expectation of the desired measurable function exists implies the existence of absolutely convergent series (see Definition~\ref{def:absolutely_convergent}), since the expectation needs to be unique in order to be well-defined.

\begin{definition}[Convergent and absolutely convergent]
\label{def:absolutely_convergent}
A series $\left(a_n\right)_{n\geq 0}$ is said to be \textbf{convergent} if, for every $\varepsilon>0$, there exist $a \in \mathbb{R}$ and $N_\varepsilon \in \mathbb{N}$ such that  $|\sum_{n=0}^m a_n - a| < \varepsilon$ for every $m\geq N_\varepsilon$.
It is said to be \textbf{absolutely convergent} if $\left(|a_n|\right)_{n\geq 0}$ converges.
\end{definition}
\noindent These definitions relate in that absolute convergence implies convergence, but the converse does not always hold.

In regards to truncation one might be able to derive stronger results and, in particular, obtain finite-iteration guarantees.
Henceforth, non-negative series that fit into a few assumptions will be the focus.
The first assumption made is that the series must pass the ratio test of convergence, meaning that
\begin{equation}
\label{eq:ratioTest}
\lim_{n \to \infty} \frac{|a_{n+1}|}{|a_n|} = L < 1.
\end{equation}
Moreover, $\left(a_n\right)_{n\geq0}$ is required to be decreasing.
In many statistical problems, particularly many p.m.f.s, the probability increases up to a mode and only then begins to decrease.
However, it is only necessary that the assumption \emph{eventually} holds, since the sum can be decomposed as
\begin{equation}
\label{eq:inf_sum_decomp}
\sum_{n=0}^\infty a_n = \sum_{n=0}^{n_0 - 1} a_n + \sum_{n=n_0}^\infty a_n,
\end{equation}
\noindent where $n_0$ is such that $a_{m+1}/a_m < 1$ for all $m \geq n_0$.
One can then define $a_n^\prime = a_{n+n_0}$ and the sum of interest becomes $S = \sum_{n=0}^{n_0} a_n + \sum_{n=0}^\infty a_n^\prime$, with truncation error arising only in the approximation of the second (``tail'') sum.

A main idea that will be explored in this paper is that of finding error-bounding pairs, and in what follows it will be convenient to establish Proposition~\ref{prop:infinite_series_bounds}, which is inspired by the results in~\cite{Braden1992}.
\begin{proposition}[\textbf{Bounding a convergent infinite series}]
\label{prop:infinite_series_bounds}
Let $S_n := \sum_{k = 0}^n a_k$.
Under the assumptions that $\left(a_n\right)_{n\geq 0}$ is positive, decreasing and passes the ratio test, then for every $0 \leq n < \infty$ the following holds:
\begin{equation}
\label{eq:sum_approx_decreasing}
 S_{n} + a_{n} \left( \frac{L}{1-L}\right) < S < S_{n} + a_{n} \left( \frac{1}{1-\frac{a_{n}}{a_{n-1}}}\right),
\end{equation}
if $\frac{a_{n+1}}{a_n}$ \textbf{decreases} to $L$ and 
\begin{equation}
\label{eq:sum_approx_increasing}
 S_{n} + a_{n} \left( \frac{1}{1-\frac{a_{n}}{a_{n-1}}}\right) < S < S_{n} + a_{n} \left(\frac{L}{1-L}\right),
\end{equation}
if $\frac{a_{n+1}}{a_n}$ \textbf{increases} to $L$.
\end{proposition}
\begin{proof}
See Appendix~\ref{app:sec:proofs}.
\end{proof}

The above result considers the case where the ratio is monotone, which is sufficient for the problems that will be presented here, but a more general version for cases where the position of the terms of the ratio relative to the limit is known (alternating ratio for example) can be found in Appendix~\ref{app:sec:genera_version_of_bounding_pairs}.

\section{Computing truncated sums}
\label{sec:approaches}

As mentioned in Section~\ref{sec:intro}, here we are concerned with adaptive truncation schemes in which the upper bound for summation $K$ is chosen so as to guarantee that $|S - S_K|\leq \varepsilon$.
We now discuss two approaches and evaluate their relative merits.

The first approach aims to align more closely with practical applications in libraries, providing a level of assurance.
And the second approach is based on the case that $\{a_n\}$ passes in the ratio test and $\frac{a_{n+1}}{a_n}$ is monotonic.
Despite being more restrictive, it is the approach among those that will be presented that has the best results, in terms of not evaluating more terms than necessary.

\subsection{Approach 1: Sum-to-threshold}
\label{sec:naive}

In most computer implementations of evaluating infinite sums given an epsilon (in general, $\varepsilon = 2.2\mathrm{e}{-16}$, machine epsilon) it is considered a ``good approximation'' to take $S = \sum_{i=1}^N a_i$ where $a_N < \varepsilon$ and $a_i > \varepsilon\ \forall i < N$.
Here we will present what is necessary for this approach to have controlled error.

Assume the positive series $\left(a_n\right)_{n\geq0}$ passes the ratio test with ratio limit $L$, that is, Equation \eqref{eq:ratioTest} is true for some $L < 1$.
Additionally, choose a number $M\in (L, 1)$.
Then the infinite sum $S$ is approximated up to an error $\varepsilon$ by $S_{n+1}$ if $a_{n+1}\frac{M}{1-M} < \varepsilon$ and $\frac{a_{n+1}}{a_n} \leq M$.
Note that, if $\frac{a_{n+1}}{a_n} \leq M$, then $\frac{a_{n+1}}{a_n} \leq 1$, that is, $\left(a_n\right)_{n\geq0}$ is decreasing.
This approximation is based on Proposition~\ref{prop:L_less_than_M} and can be implemented as in Algorithm~\ref{alg:naive} below.

\begin{proposition}[\textbf{Upper bound on the truncation error}]
\label{prop:L_less_than_M}
Let $S_n := \sum_{k = 0}^n a_k$.
Assume that $\left(a_n\right)_{n\geq 0}$ is positive, decreasing and passes the ratio convergence test for some $L < 1$.
Consider a number $M \in (L, 1)$.
This means that there exists a positive integer $n_0$ such that $\frac{a_{n+1}}{a_n} \leq M$ for every $n > n_0$ and the following holds:

\begin{equation}
\label{eq:naiveImplies}
a_n < \varepsilon \Rightarrow S - S_n < \varepsilon\frac{M}{1-M},
\end{equation}

\noindent for every $\varepsilon > 0$ and $n > n_0$.
\end{proposition}
\begin{proof}
See Appendix~\ref{app:sec:proofs}.
\end{proof}

To link the result from Proposition~\ref{prop:L_less_than_M} and the proposed approach, replace $\varepsilon$ in Equation \eqref{eq:naiveImplies} with $\varepsilon^\prime = \varepsilon\frac{1-M}{M}$ to yield the desired result.
The use of $a_{n+1}$ instead of $a_n$ implies no loss of generality.

\paragraph{Choosing $M$:} the choice of $M$ must be made carefully.
When we select $M$ close to 1, the value of $a_{n+1} \frac{M}{1 - M}$, which represents the error guarantee of the approximation, tends towards infinity.
On the other hand, if $M$ is close to $L$, the inequality $\frac{a_{n+1}}{a_n} \leq M,\ \forall n > n_0$ is only valid for large values of $n_0$.
For this reason, in general, we will choose $M = \frac{1 + L}{2}$ (the midpoint of the interval).

\begin{algorithm}
\caption{Adaptive truncation via Sum-to-threshold}
\label{alg:naive}
\begin{algorithmic}
\State Initialize $a_0$, $a_1$ and $n = 1$
\State Initialize $M = \frac{1+L}{2}$
\While{$a_n\geq\varepsilon\frac{1 - M}{M} \textbf{ or } \frac{a_n}{a_{n-1}}\geq M$}
\State Set $n = n + 1$
\State Evaluate $a_n$
\EndWhile
\State Return $S_n = \sum_{k = 0}^na_k$
\end{algorithmic}
\end{algorithm}

Note that this approach requires a prior evaluation of the problem, since it is valid for $n > n_0$, as we will see in the applications this $n_0$ can be very large, losing any guarantee of bounding.

\subsection{Approach 2: Error-bounding pairs}
\label{sec:bounding_pairs}

A popular technique for approximating infinite sums in Mathematics is trapping the true sum in an interval, and returning its midpoint as an estimate of the desired sum.
The goal of this section is to provide an easy to implement approach based on this classical idea.
Before discussing the approach, however, it is convenient to define the concept of the error-bounding pair.

\begin{definition}[\textbf{Error-bounding pair}]
Consider a convergent series $\left(a_n\right)_{n\geq0}$ with $S = \sum_{n=0}^\infty a_n$ and let $\left(L_n\right)_{n\geq0}$ and $\left(U_n\right)_{n\geq0}$ be decreasing sequences with $\lim_{n\to\infty} L_n = \lim_{n\to\infty} U_n = 0$ such that
\begin{equation*}
    L_n < S - S_n < U_n,
\end{equation*}
holds for all $n$.
We then call $(\{L_n\}, \{U_n\})$ an \textbf{error-bounding} pair which traps the true sum $S$ in a sequence of intervals $[S_n + L_n, S_n + U_n]$ of decreasing width.
\end{definition}

Now, assume the positive decreasing series $\left(a_n\right)_{n\geq 0}$ passes the ratio test, that is, Equation \eqref{eq:ratioTest} is true for some $L < 1$ and the ratio $r_{n} = \frac{a_{n+1}}{a_{n}}$ is monotonic.
Then, we have~\footnote{This is for the case where $r_n$ decreases to $L$. The reverse case is analogous with the positions of the limits reversed.}
\begin{align*}
    a_{n}\left(\frac{L}{1-L}\right) < S - S_{n} < a_{n} \left( 1- \frac{a_{n}}{a_{n-1}}\right)^{-1}
\end{align*}
is a bounding pair and the infinite sum $S$ is truncated up to an error $\varepsilon$ by
\begin{equation}
\label{eq:adaptSum}
S_n + \frac{a_{n}}{2}\left(\left(\frac{L}{1-L}\right) + \left(1-\frac{a_{n}}{a_{n-1}}\right)^{-1}\right),
\end{equation}
\noindent if 
\begin{equation}
\label{eq:boundsDiff}
a_{n}\left(1-\frac{a_{n}}{a_{n-1}}\right)^{-1} < 2 \varepsilon.
\end{equation}
Since Proposition~\ref{prop:infinite_series_bounds} bounds the remainder of the sum, equation~\eqref{eq:boundsDiff} ensures that the bounds will be up to a $2\varepsilon$ distance of each other.
Then Equation \eqref{eq:adaptSum} takes the middle point between the bounds, which guarantees that its expression is within $\varepsilon$ of the true sum. See the pseudo-code provided in Algorithm \ref{alg:adaptive}.

\begin{algorithm}
\caption{Adaptive truncation via Error-bounding pairs}
\label{alg:adaptive}
\begin{algorithmic}
\State Initialize $a_0$, $a_1$ and $n = 0$
\While{$a_{n} > a_{n-1} \textbf{ or } a_{n}\left(1-\frac{a_{n}}{a_{n-1}}\right)^{-1}\geq2\varepsilon$}
\State Set $n = n + 1$
\State Evaluate $a_{n}$
\EndWhile
\State Evaluate $S_{n} = \sum_{k = 0}^{n}a_k$
\State Return $S_{n} + \frac{a_{n}}{2}\left(\left(\frac{L}{1-L}\right) + \left(1-\frac{a_{n}}{a_{n-1}}\right)^{-1}\right)$
\end{algorithmic}
\end{algorithm}

\begin{theorem}[\textbf{Error-bounding pairs dominates Sum-to-threshold}]
\label{theo:bounding_vs_threshold}
    Let $\{a_n\}$ positive, decreasing and passes in the ratio test with $L < 1$.
    If $r_n = \frac{a_{n+1}}{a_n}$ is monotonic, i.e. we can apply both approaches, Bounding Pairs and Sum-To-Threshold.
    Then the Bounding Pairs uses fewer terms in comparison with the Sum-To-Theshold.
\end{theorem}

\begin{proof}
    See Appendix~\ref{app:sec:proofs}.
\end{proof}

An important result of this section is that when the sequence $\frac{a_{n+1}}{a_n}$ is monotonic, meaning when both approaches can be applied, then the first one yields a smaller error in the sum.
This is stated and proven in Theorem~\ref{theo:bounding_vs_threshold}.
This result makes sense since the bounding pairs approach is more restrictive.
It suggests verifying the monotonicity of the ratio of terms before applying the approach.
In Section~\ref{sec:applications}, we will see how these approaches work in practice on common problems.

\section{Computational aspects}
\label{sec:technical}

In this section we discuss how the mathematical guarantees described in Section \ref{sec:approaches} can guide the development of a computational implementation.
The interested reader is referred to Chapter 4 in~\cite{Higham2002} and to~\cite{Rump2008} and~\cite{Rump2009} for further reading on the computational aspects of numerical stability and efficiency.
\cite{Neal2015} discusses exact summation using parallel algorithms.

\paragraph{Logarithmic scale:} it is a well established fact that computation in the log scale is more stable in the sense that there are fewer situations that lead to numerical underflow and even overflow.
This is particularly important for infinite summations, since precision loss can lead to dangerous rounding error propagation.

A common method for adding numbers available in the log scale is known as the log-sum-exp algorithm.
Let $(x_1, \ldots, x_n)$ be $n$ positive numbers and $(l_1, \ldots, l_n)$ their respective natural logarithms.
Also, let $l_{(n)}$ be the latter's largest value.
Then:
\begin{equation}
\begin{aligned}
\log\hspace{0.1cm}\sum_{i = 1}^nx_i & = \log\hspace{0.1cm}\sum_{i = 1}^n \exp(l_i)\\
& = l_{(n)} + \log\left\{1 + \sum_{\substack{i = 1\\i \neq (n)}}^n\exp\left(l_i - l_{(n)}\right)\right\}.
\end{aligned}
\end{equation}

\noindent Conveniently, the function {\small\texttt{log1p(x)}} which computes $\log (1 + x)$ is implemented in a stable manner in most mathematical libraries and can be readily employed to implement the log-sum-exp technique.
This trick helps to greatly reduce precision loss in such summations.

\paragraph{Kahan summation:} when summing many terms in floating point precision, one should be careful to avoid cancellation errors.
There are many compensated summation algorithms that attempt to avoid catastrophic cancellation.
The so-called Kahan summation algorithm~\citep{Kahan1965} is one such technique that allows one to compute long sums with minimal round-off error. 
In our implementations we have taken advantage of Kahan summation, the benefits of which are summarised in Theorem~\ref{thm:kahan}.
First, however, it is convenient to define the condition number of a sum (Definition~\ref{def:cond_number}).
\begin{definition}[Condition number]
\label{def:cond_number}
For a sum $S := \sum p_i \neq 0$, the condition number is defined by
\begin{align*}
    \operatorname{cond}(S) &:= \limsup_{\varepsilon \to 0} \left\{ \left| \frac{\sum \tilde{p}_i - \sum p_i}{\varepsilon \sum p_i}\right| : |\tilde{p}| \leq \varepsilon|p| \right\},\\
    & = \frac{\sum p_i}{\left| \sum p_i\right|},
\end{align*}
where we take the absolute values and comparisons element-wise.
\end{definition}
All of the discussion in this paper centres around computing sums of non-negative terms, i.e., for which $\operatorname{cond}(S)=1$.
Now we are prepared to state
\begin{theorem}[Kahan summation algorithm]
\label{thm:kahan}
Consider computing $S_N = \sum_{n=0}^N x_n$. 
\begin{center}
\begin{algorithmic}
\State $\tilde{S} \gets x_0$
\State $C \gets 0$
\For{$j = 1$ to $N$} 
\State $Y \gets x_j - C$
\State $T \gets \tilde{S} + Y$
\State $C \gets (T-\tilde{S}) - Y$
\State $\tilde{S} \gets T$
\EndFor 
\end{algorithmic}    
\end{center}
then 
\begin{equation}
    \frac{|S_N - \tilde{S}|}{|S_N|} \leq \left[2\delta + O(N\delta^2)\right]\operatorname{cond}(S_N),
\end{equation}
where $\delta$ is the machine precision.
\end{theorem}
\begin{proof}
See the Appendix of~\cite{Goldberg1991} and the discussion of Theorem 8 therein.
\end{proof}
The discussion in Section 4.3 of \cite{Goldberg1991} is particularly helpful.
A 64-bit system has $\delta = 2^{-53} \approx 10^{-16}$.
In particular, here we shall use the maximum double precision in the R software \citep{R}, which is $\delta = 2.2 \times 10^{-16}$, and can be called with the command \texttt{\small .Machine\$double.eps}.
In Python~\citep{python312} we use the \texttt{mpmath} library~\citep{Johansson2023} to use a precision higher than machine precision, with respective machine-epsilon.

\paragraph{Ordering:} an additional trick to reduce floating point precision loss can be derived from the Kahan summation described above.
The problem that this addresses arises mainly from adding two numbers with very different orders of magnitude.
Thus the mantissa of the smallest number will be rounded off so it can be added to the largest one.
One way to mitigate this problem is to add numbers whose orders of magnitude are closest.
This can be done when summing numbers in a vector, as it is best to always add the smallest ones first, then the larger ones.
This way the smallest numbers are fully accounted for before being rounded off.
Therefore the best practice is to order the numbers in ascending order before performing the summation.

\paragraph{When $a_0$ is relatively large:} the series discussed in this paper increase up to the maximum, then enter a decreasing regime where one can employ the truncation methods with provable guarantees.
This means that one has, in principle, to check whether the maximum has been achieved before employing the truncation.
But if $a_0$ is already larger than the desired $\varepsilon$, there is no need to check whether the maximum is reached, since $a_n$ will never be smaller than $\varepsilon$ between $a_0$ and the mode.
This can ease computations as there is no need to find the maximum before checking for convergence.

\paragraph{Value for cost:} the two methods discussed in Section \ref{sec:approaches} have different implementation idiosyncrasies.
The Sum-to-threshold method is clearly the most straightforward, but it is not better than the Bounding pairs method, and its applicability is restricted to the evaluation of $n_0$.

The bounds provided by the Error-bounding pairs method are useful because not only does it have better results, but \cite{Braden1992} also argues that it provides faster convergence in the sense that it requires fewer iterations and therefore fewer function evaluations.




It is possible to try to figure out how many iterations are needed for the Sum-to-threshold and Error-bounding pairs methods before performing any function evaluations.
This way, one could make use of vectorisation for these procedures.
In particular, the respective roots of the convergence checking inequalities can indicate how many iterations must be done.
Testing has shown that while this does in fact provide the correct answer, finding the root is slower than simply checking for convergence at every step.

\paragraph{Unavoidable errors:} despite the possible numerical treatments described above, the most one can do is try to minimize them, not remove them completely.
There is still the issue that a number with infinite precision cannot be exactly represented by a computer.
Consequently, it is still possible to come across examples where the requested error $\varepsilon$ is not reached despite the mathematical guarantees discussed in Section \ref{sec:approaches}.
This can happen especially when $\varepsilon$ is close to the computer's floating point representation limit.

Upon testing the implementations from Section \ref{sec:approaches} with infinite sums whose exact values are known (see Appendix~\ref{app:sec:tests}), there have been cases where the algorithm has correctly reached the stopping point for $\varepsilon = \delta =  2.22\times10^{-16}$, but the resulting summation had an error of order $10^{-13}$.
This is not an issue with the methods themselves but with floating point representation.
More computational tricks can be implemented, but it will always be possible to find further failing examples.
This problem can in principle be addressed in some programming languages by extending numerical precision, but this requires knowledge of the needs of each problem due to the increased computational burden of increasing precision and must be evaluated accordingly.

\section{Statistical applications: Theory 
}
\label{sec:applications}

Before moving on to  test the proposed truncation schemes empirically, we discuss some types of statistical problems where they might be useful, giving theoretical guarantees where possible.
Some of the examples discussed in this section are used for the supplementary tests provided in Appendix~\ref{app:sec:tests}.

\subsection{Normalising constants}
\label{sec:norm_consts}

The first class of problems we would like to consider is computing the normalising constant for a probability mass function.
Let $X$ be a discrete random variable with support on $\mathbb{N} \cup \{0 \}$ and let $\tilde{p} : \mathbb{N}\cup \{ 0\} \to (0, \infty)$ be an unnormalised p.m.f. associated with $P$ such that
$$
\operatorname{Pr}(X = x) = \frac{1}{Z}\tilde{p}(x), x = 0, 1, \ldots,
$$
and
$$
Z := \sum_{n=0}^\infty \tilde{p}(n).
$$
While many p.m.f.s pass the ratio test, not all of them do.
For some p.m.f.s the limit of consecutive terms, $L$, can be exactly 1, which means that the ratio test for $a_n = \tilde{p}(n)$ is inconclusive.
This is the case for a p.m.f. of the form
\begin{equation}
\label{eq:weirdp.m.f.}
\tilde{p}(n) = \frac{1}{n + 1} - \frac{1}{n+2}, \quad n = 0, 1, \ldots
\end{equation}
The expression in~\eqref{eq:weirdp.m.f.} does sum to 1, but the ratio $a_{n+1}/a_n = (n + 1)/(n + 3)$ converges to $L = 1$.
Another example of a p.m.f. that does not pass the ratio test is the Zeta distribution, for which $\tilde{p}(n) = n^{-s}$, for $s \in (1, \infty)$.
Nevertheless, this is not the case for most p.m.f.s, as the distribution needs to have extremely heavy tails in order for it to yield an inconclusive ratio test.

Consider the unnormalised probability mass function of the Double Poisson distribution~\citep{Efron1986}:
\begin{equation*}
\operatorname{Pr}\left(Y  = y \mid \mu, \phi \right) \propto \tilde{p}_{\mu, \phi}(y) =   \frac{\exp(-y)y^y}{y!}\left(\frac{\exp(1)\mu}{y}\right)^{\phi y},
\end{equation*}
for $\mu, \phi >0$.
The problem at hand is to compute $K(\mu, \phi) := \sum_{n=0}^\infty \tilde{p}_{\mu, \phi}(n)$ with controlled error.
We show in Appendix~\ref{app:sec:computing_L} that for this example, $L=0$, and thus the Sum-to-threshold and Bounding pairs approaches would be suitable.
Many p.m.f.s with unknown normalising constants have $L=0$; see the Conway-Maxwell Poisson distribution in Section~\ref{sec:experiments} for another example.

For $a > 1$, consider a probability mass function (p.m.f.) of the form

\[
\tilde{p}_a(x) = \frac{1}{(x+1)^2 a^{x+1}},
\]

where $L = 1/a$. Since $a > 1$, the p.m.f. satisfies the ratio test, and the ratio $\frac{a_{n+1}}{a_n}$ is increasing. This setup aligns with the assumptions discussed previously. 
To illustrate the difference between the two approaches in an asymptotic scenario, let us consider the case where $a = 2$. In this case, using the Sum-to-threshold method with a tolerance of $\varepsilon = 2.2 \times 10^{-16}$, we need 42 terms to ensure that $|S - S_n| \leq \varepsilon$. In contrast, the Error-bounding pairs method requires only 37 terms to achieve the same level of accuracy. As $a$ approaches to 1, making the problem more difficult, the Error-bounding pairs method proves to be more effective, as demonstrated in Table~\ref{tab:stt_vs_bp}.

\subsection{Raw and factorial moments of discrete random variables}
\label{sec:moments}

Another large class of statistical problems involves computing moments, which can be used in generalised method of moments estimation~\citep{Hall2004}, for example.
A brief discussion is presented on the applicability of the methods developed here to the problem of computing moments of random variables defined with respect to discrete probability distributions.

First, in Remark~\ref{rmk:factmom} it is shown that computing raw and factorial moments is amenable to the techniques developed here under the assumption that the p.m.f. passes the ratio test.
In summary, as long as the p.m.f. passes the ratio test, one will be able to use the methods developed here to accurately compute moments with guaranteed truncation error.

\begin{remark}[\textbf{Approximating raw and factorial moments}]
\label{rmk:factmom}
Let $X$ be a discrete random variable with support on $\mathbb{N}\cup \{0\}$ with distribution $\mathcal{M}(\boldsymbol{\theta})$ and p.m.f. given by $\operatorname{Pr}(X =x) = f(x \mid \boldsymbol{\theta})$.
Suppose one is interested in either raw ($E_{\mathcal{M}}[X^r]$) or factorial ($E_{\mathcal{M}}[(X)_r]$) moments, for some order $r\geq1$.
If we have $\lim_{n \to \infty} f(x+1\mid \boldsymbol{\theta})/f(x\mid \boldsymbol{\theta}) = L < 1$, then one can use the Sum-to-threshold to approximate $E_{\mathcal{M}}[X^r]$ or $E_{\mathcal{M}}[(X)_r]$ to a desired accuracy $\varepsilon > 0$.
Additionally, if $f(x+1\mid \boldsymbol{\theta})/f(x\mid \boldsymbol{\theta})$ is also decreasing in x, we can use the Bounding-pairs method.
\end{remark}
\begin{proof}
See Appendix~\ref{app:sec:proofs}.
\end{proof}
The example in Section~\ref{sec:comp_noisy_mcmc}, while relating to the computation of a normalising constant, can also be seen as that of computing a factorial moment.

\subsection{Marginalisation: the case of count models with observation error}
\label{sec:obs_error}

Marginalising out discrete variables is crucial for algorithms such as dynamic Hamiltonian Monte Carlo (dHMC) employed in Stan~\citep{Carpenter2017}, which rely on computing gradients with respect to all random quantities in the model and thus cannot handle discrete latent quantities directly.
Hence, marginalisation constitutes an essential class of infinite series-related problems.
It is straightforward to show that marginalisation is particularly amenable to the techniques discussed here (Proposition~\ref{prop:marginal}).

\begin{proposition}[\textbf{Marginalisation and the ratio test}]
\label{prop:marginal}
Let $X$ be a discrete random variable with support on $\mathbb{N} \cup \{0 \}$ and $Y$ be any measurable function that leads to some probability space $(\Omega, \mathcal{A}, \mathcal{P})$, so that a joint probability space for $X$ and $Y$ is well-defined as well.
For any $E\in\mathcal{A}$, define marginalisation as the operation
\begin{equation}
\mathcal{P}(E) = \sum_{n=0}^\infty \operatorname{Pr}(X=n, Y\in E),
\end{equation}
that is, the marginal distribution of $Y$ is the sum over all possible values of $n$ of the joint distribution of $X$ and $Y\in E$.
Denote $\operatorname{Pr}(X = n\mid Y\in E)$ as the conditional probability function of $X$ given $Y\in E$, for all $E\in \mathcal{A}$ with positive measure. If this p.m.f passes the ratio test a.s., that is, if
\begin{equation}
\lim_{n\rightarrow\infty}\frac{\operatorname{Pr}(X=n+1\mid Y\in E)}{\operatorname{Pr}(X=n\mid Y\in E)} = L(E),
\end{equation}
\noindent for some $L(E)<1$ and almost all $E\in \mathcal{A}$ for which $\mathcal{P}(E) > 0$, then the marginalisation operation also passes the ratio test a.s.
\end{proposition}
\begin{proof}
See Appendix~\ref{app:sec:proofs}.
\end{proof}

Now some illustrations are presented on how adaptive truncation might be employed for marginalisation by discussing count models with observation error, which have applications in Ecology and Medicine.
A common situation is modelling (e.g. disease) cluster sizes with a random variable $Y_i$, which represents the size (number of individuals in) a cluster.
Suppose that we observe each individual with probability $\eta$ but if any individual in the cluster is observed, the whole cluster is observed.
This is the so-called size-dependent or sentinel detection model.
It is common, for instance, in quality control and disease contact-tracing  settings~\citep{Blumberg2013b}. 
The model can be formulated as

\begin{align*}
 Y_i &\sim \mathcal{M}(\boldsymbol\theta), \: i = 0, 1, 2, \ldots, K, \\
 Z_i \mid Y_i & \sim \operatorname{Bernoulli}\left( 1-(1-\eta)^{Y_i} \right), \\
 X_i \mid Z_i &= \begin{cases}
	  0, Z_i = 0,\\
	  Y_i, Z_i = 1.
	 \end{cases} 
\end{align*}

Here $\mathcal{M}$ is a discrete distribution with unbounded support, indexed by parameters $\boldsymbol{\theta}$.
The actual observed data is represented by the random variable $X_i$, whilst the $Z_i$ and $Y_i$ are latent. 
We can then write

\begin{align}
\label{eq:p_of_X_sentinel}
\rho_0 &:= E_{\mathcal{M}}[(1-\eta)^{Y_i}] = \sum_{n = 0}^\infty \operatorname{Pr}(Y_i = n  \mid  \boldsymbol\theta)(1-\eta)^n,\\	
\nonumber
 \operatorname{Pr}(X_i = x \mid  \boldsymbol\theta, \eta) &= \begin{cases}
	  \rho_0, x = 0,\\
	  \operatorname{Pr}(Y_i = x  \mid  \boldsymbol\theta)\left(1-(1-\eta)^x\right), x > 0.
	 \end{cases} 
\end{align}

Since when $X_i = 0$ we do not observe data, i.e. detect the cluster, we need to actually model the zero-truncated random variable $X_i^\prime$.
Under zero-truncation, we can write  the p.m.f. of $X_i^\prime$ as
\begin{equation*}
 \label{eq:p_of_X_prime_sentinel}
  \operatorname{Pr}(X_i^\prime = x^\prime) = \frac{\operatorname{Pr}(Y_i = x^\prime  \mid  \boldsymbol\theta)\left(1-(1-\eta)^{x^\prime}\right)}{1-\rho_0}.
\end{equation*}
And the first moment of $X_i^\prime$ is
\begin{align*}
 \label{eq:expectation_X_prime_sentinel}
 E[X_i^\prime] &= \sum_{n=0}^\infty (n + 1) \frac{\operatorname{Pr}(Y_i = n + 1 \mid  \boldsymbol\theta)\left(1-(1-\eta)^{(n+1)}\right)}{1-\rho_0},\\
 &= \frac{E_{\mathcal{M}}\left[Y_i\right]-E_{\mathcal{M}}\left[Y_i(1-\eta)^{Y_i}\right]}{1-\rho_0}, 
\end{align*}
which is itself dependent on also computing the infinite sum $\sum_{n = 0}^\infty n \operatorname{Pr}(Y_i = n  \mid  \boldsymbol\theta)(1-\eta)^n$.

This formulation is thus contingent on $\rho_0$ being easy to compute, preferably in closed-form.
When $\mathcal{M}$ is a Poisson distribution with rate $\lambda$, we know that $\rho_0 = e^{-\lambda\eta}$ and when it is a negative binomial distribution with mean $\mu$ and dispersion $\phi$, we have $\rho_0 =  \left[ \frac{\phi}{ \eta\mu + \phi }  \right]^\phi$.
In both cases, $L = 0$ -- see Remark~\ref{rmk:marg_count_error}.
In Appendix~\ref{app:sec:tests} we exploit an example where the closed-form solution is known in order to evaluate the proposed truncation schemes for computing $\rho_0$ as given in (\ref{eq:p_of_X_sentinel}).

The size-independent or binomial model of observation error is also a very popular choice, finding a myriad of applications in Ecology (see, e.g.~\cite{Royle2004}).
The model reads
\begin{align*}
 Y_i &\sim \mathcal{M}(\boldsymbol\theta), \: i = 0, 1, 2, \ldots, K, \\
 X_i \mid Y_i & \sim \operatorname{Binomial}(Y_i, p).
\end{align*}

Next a specialisation of Proposition~\ref{prop:marginal} is provided for count models with both size-dependent and (binomial) size-independent observation error, summarised in Remark~\ref{rmk:marg_count_error}.
\begin{remark}[\textbf{Marginalisation in both size-dependent and size-independent observation error models passes the ratio test}]
Under the assumption that the pmf passes the ratio test and the moments exist, one can show that both raw and factorial moments also pass the ratio test and thus can be approximated using Sum-to-threshold or with Error-bounding pairs if the ratio of terms is monotonic.
\label{rmk:marg_count_error}
\end{remark}
\begin{proof}
See Appendix~\ref{app:sec:proofs}.
\end{proof}

In addition to the computational advantages of being able to use efficient algorithms to fit models that otherwise would not be tractable, marginalisation  might also provide improved \textbf{statistical} efficiency due to being a form of Rao-Blackwellisation~\citep{Robert2021}.
See Section 6 of~\cite{Pullin2020} for more discussion on the statistical benefits of marginalisation.

\section{Statistical applications: Illustrations}
\label{sec:experiments}

Now that common statistical applications of (adaptive) truncation methods have been discussed, in this section some fully worked out empirical examples are provided where adaptive truncation can be employed to improve the computational aspects of important statistical applications.
These cover from normalising constants to marginalisation in maximum likelihood estimation.

\subsection{Noisy MCMC for the Conway-Poisson distribution}
\label{sec:comp_noisy_mcmc}

We start our investigation with an example of Markov chain Monte Carlo (MCMC) with a noisy approximation of the likelihood.
The Conway-Maxwell Poisson distribution (COMP, \cite{Conway1962}) is a popular model for count data, mainly due to its ability to accommodate under- as well as over-dispersed data - see \cite{Sellers2012} for a survey.
For $\lambda > 0$ and $\nu > 0$, the COMP probability mass function (p.m.f.) can be written as
\begin{equation*}
    p_{\lambda, \nu}(n) := \operatorname{Pr}(X = n \mid \lambda, \nu) = \frac{\lambda^n}{Z(\lambda, \nu) (n!)^\nu},
\end{equation*}
where 
\begin{equation}
    \label{eq:COMP_normalising}
    Z(\lambda, \nu) := \sum_{n=0}^\infty \frac{\lambda^n}{(n!)^\nu}
\end{equation}
is the normalising constant.
The sum in \eqref{eq:COMP_normalising} is not usually known in closed-form for most values of $(\lambda, \nu)$ and thus needs to be computed approximately.
Notable exceptions are $Z(\lambda, 1) = \exp(\lambda)$ and $Z(\lambda, 2) = I_0(2\sqrt{\lambda})$, where $I_0$ is the modified Bessel function of the first kind -- see Section~\ref{sec:mmle_erlang} below.
While custom approximations have been developed for the COMP normalising constant~\citep{Gaunt2019}, these usually do not guarantee that one is able to compute the approximate normalising constant within a given tolerance.
We thus employ the methods developed here to consider approximations with a guaranteed approximation error; in this case the Sum-to-threshold approach with threshold $\varepsilon$ delivers the approximate sum within $\varepsilon$ tolerance, as summarised in Remark~\ref{rmk:comp_naive}.

\begin{remark}[\textbf{Error-bounding pairs truncation for the Conway-Maxwell Poisson}]
\label{rmk:comp_naive}
A error-bounding pairs truncation scheme will yield an approximation of $Z(\lambda, \nu)$  within $\varepsilon >0$, i.e. $|Z(\lambda, \nu)-\sum_{n=0}^{n^\star} p_{\lambda, \nu}(n)| \leq \varepsilon$, for all $n^\star$ such that $p_{\lambda, \nu}(m) \left( 1 - \frac{p_{\lambda, \nu}(m)}{p_{\lambda, \nu}(m-1)}\right)^{-1} \leq \varepsilon$ for all $m \geq n^\star$.
\end{remark}
\begin{proof}
See Appendix~\ref{app:sec:proofs}.
\end{proof}

Consider the situation where one has observed some independent and identically distributed (i.i.d.) data $\boldsymbol{y}$ assumed to come from a COMP distribution with parameters $\lambda$ and $\nu$ and one would like to obtain a posterior distribution $p(\lambda, \nu \mid \boldsymbol{y}) \propto f(\boldsymbol{y} \mid \lambda, \nu)\pi(\lambda, \nu)$.
This Bayesian inference problem constitutes a so-called doubly-intractable problem, because neither the normalising constant of the posterior $p(\lambda, \nu \mid \boldsymbol{y})$ nor that of the likelihood $f(\boldsymbol{y} \mid \lambda, \nu)$ are known.
In some situations, one may bypass computing $Z(\lambda, \nu)$ entirely, but this entails the use of specialised MCMC algorithms~\citep{Benson2021}.
An alternative to the rejection-based algorithm of~\cite{Benson2021} are the so-called noisy algorithms, where the likelihood is replaced by a (noisy) estimate at every step of the MCMC~\citep{Alquier2016}.
Our adaptive truncation approaches integrate seamlessly into this class of algorithms, with the added benefit that an approximation of the likelihood with controlled error will yield an algorithm where in principle one can make the approximation error negligible compared to Monte Carlo error.

Following the discussion of noisy algorithms with fixed summation bound K given in~\cite{Benson2021}, we will provide evidence that adaptive truncation yields very good results.
First, however, we need to introduce a reparametrisation of the COMP employed by the authors that facilitates its application in generalised linear models (here $\lambda = \mu^\nu$):
\begin{equation*}
    \tilde{p}_{\mu, \nu}(n) = \frac{\mu^{\nu n}}{\tilde{Z}(\mu, \nu) (n!)^\nu},
\end{equation*}
where 
\begin{equation}
    \label{eq:COMP_normalising_2}
    \tilde{Z}(\mu, \nu) := \sum_{n=0}^\infty \left(\frac{\mu^n}{n!}\right)^\nu.
\end{equation}

Note that both (\ref{eq:COMP_normalising}) and (\ref{eq:COMP_normalising_2}) series have $\frac{a_{n+1}}{a_n}$ monotonic, so we can take the Error-bounding pair approach.

In their analysis of the inventory data\footnote{Taken from~\url{http://www.stat.cmu.edu/COM-Poisson/Sales-data.html}.} of ~\cite{Shmueli2005}, \cite{Benson2021} place a Gamma(1, 1) prior on $\mu$ and Gamma(0.0625, 0.25) prior on $\nu$, a suggestion we will follow here.
The authors mention that fixing $K = 100$ or $K=3,300$ yields identical results\footnote{The value $K=3300$ was chosen by~\cite{Benson2021} so as to make the noisy algorithm have comparable runtime to their rejection sampler.}, but correctly point out that if the sampler starts out at a point with extreme values such as $(\mu = 500, \nu=0.001)$, it might fail to converge because many more iterations than $K$ are needed to approximate $\tilde{Z}(\mu, \nu)$ to a satisfactory tolerance.
In their Figure 5,~\cite{Benson2021} show that for some values of $\mu$ and $\nu$ the approximation will take many more than 1000 iterations.
In Table~\ref{tab:COMP_iters}, we leverage the techniques developed here to provide the exact numbers of iterations needed to achieve a certain tolerance $\varepsilon$ using approaches 2 and a version of approach 1 in which the real value is known (obtained with a very large number of terms) called here Sequential, for the same parameter values considered by~\cite{Benson2021}.
We also show the errors of two R libraries -- \textbf{brms}~\citep{Burkner2018} and \textbf{COMPoissonreg}~\citep{Sellers2023} -- for the respective parameters, and how their error is above that requested for bounding pairs.

\begin{table}[!ht]
\centering
\begin{tabular}{@{}ccccccc@{}}
\toprule
            & \multicolumn{2}{c}{$\varepsilon = 2.2 \times 10^{-10}$} & \multicolumn{2}{c}{$\varepsilon = 2.2 \times 10^{-16}$} & \multicolumn{2}{c}{} \\ \midrule
  Parameters & Sequential & Bounding & Sequential & Bounding & brms & COMPoissonReg\\
\hline
$\mu = 10^1$, $\nu = 10^{-1}$ & 141 & 139 & 190 & 189 & 1.9569e-14 & 2.4317e-05 \\
$\mu = 10^2$, $\nu = 10^{-2}$ & 1506 & 1482 & 1986 & 1964 & 7.4670e-13 & 4.0071e+01 \\
$\mu = 10^3$, $\nu = 10^{-3}$ & 15907 & 15662 & 20637 & 20411 & 1.1504e-10 & 4.1210e+02\\
$\mu = 10^4$, $\nu = 10^{-4}$ & 167275 & 164854 & 213910 & 211671 & 1.7542e-08 & 4.1368e+03\\ \bottomrule
\end{tabular}
\caption{\textbf{Numbers of iterations needed to approximate the normalising constant of the COMP and error of R libraries}.
We show the number $n$ of iterations needed to obtain $|\tilde{Z}(\mu, \nu) - \sum_{x=0}^n\tilde{p}_{\mu, \nu}(x)| \leq \varepsilon$ for $\varepsilon = \delta$ and $\varepsilon = 10^6 \delta$, where $\delta$ is machine precision (given in R by .Machine\$double.eps).
Results for the Sequential (version of Sum-to-threshold with an approximation of the result) and Error-bounding pair approaches are provided.}
\label{tab:COMP_iters}
\end{table}

In addition, we analyse the inventory data using implementations of the sum-to-threshold~\footnote{Although it fails the ratio test and the sum-to-threshold is valid, for this problem $n_0$ is greater than the number of necessary iterations, that is, sum-to-threshold is not guaranteed for this case.} (approach 1) and error-bounding pair (approach 2) truncation algorithms in the Stan~\citep{Carpenter2017} programming language -- please see Appendix~\ref{app:sec:comp_dets} for details -- and reveal why~\cite{Benson2021} find that $K=100$ is sufficient for the analysis of these data: the median number of iterations needed to approximate the normalising constant to within $\varepsilon \approx 2.2 \times  10^{-16}$ of the truth was around $80$ for approaches 1 and 2.
The results of this analysis are given in Table~\ref{tab:COMP_inventory_results} and show that fixing $K$ is not a good approach in terms of Effective Sample Size (ESS)/minute, the adaptive truncation algorithms perform fewer iterations and achieve very satisfactory performance without burdening the analyst with having to choose $K$.

\begin{table}[]
\centering
\begin{tabular}{@{}ccccccc@{}}
\toprule
                                &              & Posterior median (BCI)    & Posterior sd    & MCSE  & ESS/minute \\ \midrule
\multirow{3}{*}{Threshold}      & $\mu$           & 0.805 (0.530, 1.087) & 0.142 & 0.003 & 126740      \\
                                & $\nu$           & 0.127 (0.104, 0.150) & 0.012 & 0.000     & 127957   &         \\
                                & $n$ & 80 (75, 86)          & 2.901  & 0.048 & 142676                   \\
                                \midrule
\multirow{3}{*}{Error-bounding pair } & $\mu$           & 0.803 (0.533, 1.073)   & 0.138 & 0.003 & 75970   \\
                                & $\nu$           & 0.127 (0.105, 0.149) & 0.011 & 0.000     & 75871    &         \\
                                & $n$ & 81 (76, 88)          & 2.931 & 0.050 & 85034                   \\
                                \midrule
\multirow{2}{*}{Fixed K = 100}  & $\mu$           & 0.800 (0.519, 1.074) & 0.140 & 0.002 & 124857  \\
                                & $\nu$           & 0.127 (0.104, 0.150) & 0.012 & 0.000     & 123964            \\
\multirow{2}{*}{Fixed K = 3300} & $\mu$           & 0.805 (0.536, 1.087) & 0.140 & 0.002 & 5052    \\
                                & $\nu$           & 0.127 (0.105, 0.150) & 0.012 & 0.000 & 5055 \\ \bottomrule 
\end{tabular}
\caption{\textbf{Bayesian analysis of inventory data~\citep{Shmueli2005} under Conway-Maxwell Poisson model using noisy algorithms}.
We show the posterior mean and Bayesian credible interval (BCI) for $\mu$ and $\nu$ and the median number of iterations $n$ needed to get an approximation within $\varepsilon = 2.2 \times 10^{-16}$ of the true normalising constant.
Results for the noisy algorithm with fixed $K$ as discussed in~\cite{Benson2021} are also given for comparison.
We provide estimates of the Monte Carlo standard error (MCSE) and effective sample size (ESS) per minute.
}
\label{tab:COMP_inventory_results}
\end{table}

\subsection{Maximum marginal likelihood in a toy queuing model}
\label{sec:mmle_erlang}

We now move on to study the application of adaptive truncation to marginalisation problems, and choose maximum marginal likelihood estimation (MMLE) as our example.
Consider a very simple queuing model where a (truncated) Poisson number of calls ($Y$) are made and call duration follows an exponential distribution with rate $\beta$.
Moreover, only the total duration of all calls, $X_i$, is recorded.
The model can be written as
\begin{align}
\nonumber
    Y_i &\sim \operatorname{Truncated-Poisson}(\mu; 0), i = 1,\ldots, J,\\
    \nonumber
    Z_j & \sim \operatorname{Exponential}(\beta), j = 1, \ldots, Y_i,\\
    \label{eq:erlang_model}
    X_i &= \sum_{j=1}^{Y_i} Z_j.
\end{align}
It is well-known that the $X_i$ follow the so-called Erlang distribution, i.e., a Gamma distribution where the shape parameter is an integer. 
Our goal is to make inference about $\theta = (\mu, \beta)$ from a collection of i.i.d. observations $\boldsymbol{x} = \{x_1, \ldots, x_J\}$.

In particular, the goal is to maximise the marginal likelihood $L(\theta \mid \boldsymbol{x}) = \prod_{i=1}^J f_X(x_i \mid \mu, \beta)$, in  order to find the maximum marginal likelihood estimate $\hat{\theta} = (\hat{\mu}, \hat{\beta})$.
To this effect, we compute 
\begin{align}
   \label{eq:erlang_marg_like}
    f_X(x \mid \mu, \beta) &= \sum_{n=1}^\infty \operatorname{Pr}(Y=n \mid \mu)f_{X \mid Y}(x \mid Y = n, \beta),\\
    \label{eq:erlang_full_rep}
    &= \sum_{n=1}^\infty \frac{\exp\left(-(\mu + \beta x)\right)}{(1-\exp(-\mu))x} \frac{\left(\mu x \beta\right)^n}{n!(n-1)!},\\
    \label{eq:erlang_bessel_rep}
    &= \frac{\exp\left(-(\mu + \beta x)\right)}{(1-\exp(-\mu))x}\sqrt{\mu\beta x}\cdot I_1\left(2\sqrt{\mu\beta x}\right),
\end{align}
where
\begin{equation*}
    I_v(z) = \left(\frac{z}{2}\right)^v\sum_{k=0}^\infty \frac{\left(\frac{z^2}{4}\right)^k}{k!\Gamma(v + k +1)}, 
\end{equation*}
for $v, z > 0$, is the modified Bessel function of the first kind.
Computation of this special function can be numerically unstable, specially when $\mu$ gets large.
The techniques presented in this paper allow for robust implementations of the Bessel function in log-space, and thus lead to stable computation of the marginal log-likelihood.

In order to study whether adaptive truncation provides an advantage compared to using a fixed truncation bound, we devised a simulation experiment: for a pair of data-generating parameter values $(\mu, \beta)$, we generate $500$ data sets with $J=50$ data points each from the model in (\ref{eq:erlang_model}).
Then, for each data set we found the MLE by maximising the marginal likelihood in (\ref{eq:erlang_marg_like}) by computing either (\ref{eq:erlang_full_rep}), which we will henceforth call the `full' representation or (\ref{eq:erlang_bessel_rep}), which we shall call the Bessel representation.
The optimisation was carried out using the \verb|mle2()| routine of the \textbf{bbmle} package~\citep{Bolker2020}, which implements the Limited-memory Broyden-Fletcher-Goldfarb-Shanno (L-BFGS) algorithm~\citep{Liu1989}.

Interestingly, these two representations do not lead to the same number of iterations under adaptive truncation, with the full representation usually needing fewer iterations to reach the stopping criteria.
We thus exploit these differences in order to understand how they relate to statistical efficiency and numerical accuracy.
For each representation, we compute the MLE using either fixed ($K=1000$) or adaptive truncation by the Error-bounding pairs method\footnote{Since $L=0$, this method is guaranteed to give an approximation of controlled error. Note that, if $L=0$, then the ratio $\frac{a_{n+1}}{a_n}$ is monotonic.}.
Interval estimates in the form of approximate 95\% confidence intervals (CIs) were also computed.
These employ the well-known asymptotic Delta method, and depend on obtaining the Hessian matrix evaluated at the MLE, and in modern implementations, it can be approximated numerically by numerical differentiation. 
Since this procedure can be numerically unstable, especially for higher-order derivatives, we implement the Hessian directly as derived in Appendix~\ref{app:sec:hessian}.
Notice that the computations of the analytic Hessian also rely on adaptive truncation.

Table~\ref{tab:mmle_erlang_mu} shows the average computing time in seconds, the root mean squared error (RMSE) and the coverage of the 95\% CIs under both numerical and analytic (exact) implementations and under both representations. 
The smaller number of iterations needed to reach convergence for the full representation do lead to an overall decreasing in computing time across experimental designs and truncation strategies.
However, the Bessel representation appears to lead to more numerically stable computation, as can be seen by the coverage of its numerical differentiation-based CIs for the case with $\mu=1500$ being close to nominal, whilst the coverage attained by the full representation is substantially lower ($0.65$).

Importantly, the results clearly show that while for some configurations of the data-generating process the results were indistinguishable between fixed and adaptive truncation, for $\mu=1500$ using a fixed cap lead to an RMSE that was twice that of the adaptive approach (for both representations) and coverage that was disastrously low -- none of the estimated CIs were able to trap the true parameter values.
Moreover, the adaptive approach also lead to faster computation in general for $\mu \in \{15, 150\}$, and while being slower for $\mu=1500$, it also yielded much better statistical performance.

\begin{table}[]
\centering
\begin{tabular}{@{}cccccc@{}}
\toprule
True & & Time & RMSE & \multicolumn{2}{c}{Coverage} \\ \midrule
 &  & fixed/adaptive & fixed/adaptive & Numerical & Analytic \\
\multirow{2}{*}{$\mu = 15$} & Bessel & 1.95/2.08 & 3.76/3.76 & 0.91/0.91 & 0.94/0.94 \\
 & Full & 2.81/1.37 & 3.76/3.76 & 0.91/0.91 & 0.94/0.94 \\
\multirow{2}{*}{$\mu = 150$} & Bessel & 1.86/2.33 & 32.97/32.97 & 0.96/0.96 & 0.96/0.96 \\
 & Full & 2.79/1.66 & 32.97/32.97 & 0.96/0.96 & 0.96/0.96 \\
\multirow{2}{*}{$\mu = 1500$} & Bessel & 3.02/10.65 & 579.4/355.7 & 0.00/0.94 & 0.00/0.95 \\
 & Full & 4.89/8.23 & 580.3/355.6 & 0.00/0.65 & 0.00/0.95 \\ \bottomrule
\end{tabular}
\caption{\textbf{Root mean squared error and coverage results for the $\mu$ parameter in the Erlang queuing model}.
Using $500$ replicates per design, we show the average computing time in seconds, root mean squared error (RMSE) and confidence interval coverage for the estimation of $\mu$ -- see Table~\ref{tab:mmle_erlang_beta} for the results for $\beta$.
In all experiments, the true generating $\beta = 0.1$.
All results are shown as fixed/adaptive, where the fixed implementation uses $K=1000$ iterations and the adaptive implementation uses the threshold approach.
We show the coverage of confidence intervals computed using the Hessian matrix approximated using either numerical differentiation or the analytic calculations in Appendix~\ref{app:sec:hessian}.
See text for more details.
}
\label{tab:mmle_erlang_mu}
\end{table}

\section{Discussion}
\label{sec:discussion}

Problems relying on infinite summation are ubiquitous and can be found in fields as diverse as Phylogenetics~\citep{Cilibrasi2011} and Psychology~\citep{Navarro2009}.
Whatever the application, stable and reliable algorithms are of utmost importance to ensure correctness and reproducibility of results, in particular by avoiding or controlling error propagation.
Here some techniques were proposed and analysed for the truncation of infinite sums of non-negative series by unifying the practical aspects of their implementation with analytical justification for their use.
Now, a few of the lessons learned from the efforts reported in the present paper are discussed.

\subsection{Having provable guarantees}

A major concern when implementing an algorithm is numerical stability: can one guarantee that the inevitable errors introduced by representing abstract mathematical objects in floating-point arithmetic remain under control?
Moreover, even if there are no major under/overflow or catastrophic cancellation issues, one might still want to have mathematical guarantees of correctness in the form of controlled truncation error.
This issue is even more evident in approximations within MCMC, since these methods have fragile regularity conditions -- see~\cite{Park2020} and further discussion below.

The results presented here show that so long as one can compute the limit of the consecutive terms ratio, $L < 1$, one can pick the right method to perform adaptive truncation.
Moreover, adaptive truncation lead to better efficiency by saving computation where it was not needed and higher accuracy by allowing truncation bounds to expand when necessary.
It is thus clear that if the regularity conditions described in Section~\ref{sec:prelim} are met, one is much better served by using the algorithms described here.
One must be vigilant however in checking that the correct algorithm for each summation problem is employed.
As shown in the Theorem~\ref{theo:bounding_vs_threshold}, if the ratio $\frac{a_{n+1}}{a_n}$ is monotonic.
It is better to use the error-bounding pairs approach, but it does not guarantee that the difference between the approaches is large.
As shown in the end of Section~\ref{sec:applications}, when $L > 0.5$ picking the Sum-to-threshold might lead to an approximation with higher error than required.
This is because Sum-threshold has no guarantees in this case, despite showing low error in some cases.
See, in particular, Table~\ref{tab:tests_negbinom_64bits}.
Limits are usually straightforward to evaluate, with rare cumbersome exceptions and we discuss a few techniques that might make it easier to find $L$ in algebraically complicated problems in Appendix~\ref{app:sec:computing_L}.

The methods presented here have a broad range of applicability, requiring only mild conditions be met and can be directly applied to doubly-intractable problems such as the Conway-Maxwell Poisson example in Section~\ref{sec:comp_noisy_mcmc}.
They may present an alternative to stochastic truncation techniques such as Russian Roulette~\citep{Lyne2015}, which involve replacing the upper truncation bound $K$ with a random variable for which probabilistic guarantees can be given.
Russian Roulette for example constructs a random variable $\tau_\theta$, which may depend on a set of parameters $\theta$, and the authors are able to show that this preserves unbiasedness. 
A similar technique is discussed in Section 2 of~\cite{Griffin2016} in the context of the compound Poisson process approximation to Lévy processes. 
How deterministic adaptive truncation compares to these stochastic approaches is an interesting question for future research. 

\subsection{Limitations and extensions}

Despite the desirable guarantees that the methods provide, it is natural that they under-perform, in terms of computing time, relative to custom-made methods such as asymptotic approximations~\citep{Gaunt2019} or clever summation techniques that exploit the specific structure of a problem -- see Appendix B in~\cite{Meehan2020}.
A good example is the modified Bessel function of the first kind discussed in Section~\ref{sec:experiments}.
Although the adaptive truncation algorithms are able to yield results that are comparable to custom algorithms such as the one in the \verb|besselI()| function in R, we have found that computation time is ten to fifteen times larger (data not shown).
However, since our methods are provided with explicit guarantees, they can provide a reliable benchmark for the development of faster, custom-made calculations.

Another limitation of the presented methods is that, as mild as the regularity conditions they require are, there are still interesting problems for which they are not suitable. 
A good example is computing the normalising constant when the p.m.f. in question involves a power-law term, as in~\cite{Gillespie2017}.
While custom techniques based on Euler-Maclaurin error-bounding pairs can be shown to be quite powerful in problems with $L=1$ and slowly-converging series in general~\citep{Boas1978,Braden1992}, the challenge is \textbf{algorithmisation}, i.e., being able to turn a powerful technique into a problem-agnostic algorithm that handle many problems in a broad class. 
If these Euler-Maclaurin methods can be made broadly applicable, one might be able to give good theoretical guarantees based on asymptotic bounds on the remainder~\citep{Weniger2007}.

We also do not address series for which terms can be negative, such as those which appear in first-passage time problems as discussed in e.g.~\cite{Navarro2009}.
Their inclusion, while feasible, will require further theoretical and programming work.
Moreover, while our methods are directly applicable to convergent alternating series, we do not pursue that route here.
One reason for this is that, as \cite{Kreminski1997} shows (Example 2 therein), under mild conditions on the alternating series, one can do much better than for these problems than the algorithms proposed here.
Finally, we note that as the example in Section~\ref{sec:mmle_erlang} shows, different representations of the same series can yield faster or slower converging summation problems.
These differences can thus be exploited for series acceleration (see Chapter 8 in~\cite{Small2010}), the algorithmisation of which is a worthy goal for future research.

In closing, we hope the present paper provides the statistical community with a robust set of tools for accurate and stable computation of the many infinite sums that crop up in modern statistical applications.

\section*{Code availability}

An R package implementing the methods described here is available from~\url{https://github.com/GuidoAMoreira/sumR}.
The calculations are performed at low level, that is, they are programmed in C, and new versions are uploaded to CRAN as soon as they are stable.
A Python package is also available at~\url{https://github.com/wellington36/InfSumPy}.
The \textit{mpmath} library was used for high-precision numerical evaluation.
The library is available on PyPi.
Stan code implementing the algorithms can be obtained from~\url{https://github.com/GuidoAMoreira/stan_summer} and code to implement the Conway-Maxwell Poisson in Stan is at~\url{https://github.com/wellington36/MCMC_COMPoisson}.
Scripts using \textbf{sumR} to reproduce the results presented in the paper can be found at~\url{https://github.com/maxbiostat/truncation_tests}.
And using \textbf{InfSumPy} can be found at~\url{https://github.com/wellington36/adaptive_truncation_table_generator}.

\section*{Acknowledgements}

We thank H\"avard Rue, Ben Goodrich, Alexandre B. Simas, Ben Goldstein, Hugo A. de la Cruz, Alan Benson and Jairon Batista for enlightening discussions.

\bibliography{adaptive_truncation}

\begin{thebibliography}{}

\bibitem[Aleshin-Guendel et~al., 2021]{Aleshin2021}
Aleshin-Guendel, S., Sadinle, M., and Wakefield, J. (2021).
\newblock Revisiting identifying assumptions for population size estimation.
\newblock {\em arXiv preprint arXiv:2101.09304}.

\bibitem[Alquier et~al., 2016]{Alquier2016}
Alquier, P., Friel, N., Everitt, R., and Boland, A. (2016).
\newblock Noisy {M}onte {C}arlo: Convergence of markov chains with approximate
  transition kernels.
\newblock {\em Statistics and Computing}, 26(1-2):29--47.

\bibitem[Benson and Friel, 2021]{Benson2021}
Benson, A. and Friel, N. (2021).
\newblock Bayesian inference, model selection and likelihood estimation using
  fast rejection sampling: The {C}onway-{M}axwell-poisson distribution.
\newblock {\em Bayesian Analysis}.

\bibitem[Betancourt, 2017]{Betancourt2017}
Betancourt, M. (2017).
\newblock A conceptual introduction to {H}amiltonian {M}onte {C}arlo.
\newblock {\em arXiv preprint arXiv:1701.02434}.

\bibitem[Blumberg and Lloyd-Smith, 2013]{Blumberg2013b}
Blumberg, S. and Lloyd-Smith, J.~O. (2013).
\newblock Comparing methods for estimating {$R_0$} from the size distribution
  of subcritical transmission chains.
\newblock {\em Epidemics}, 5(3):131--145.

\bibitem[Boas, 1978]{Boas1978}
Boas, R.~P. (1978).
\newblock Estimating remainders.
\newblock {\em Mathematics Magazine}, 51(2):83--89.

\bibitem[Bolker and {R Development Core Team}, 2021]{Bolker2020}
Bolker, B. and {R Development Core Team} (2021).
\newblock {\em bbmle: Tools for General Maximum Likelihood Estimation}.
\newblock R package version 1.0.24.

\bibitem[Braden, 1992]{Braden1992}
Braden, B. (1992).
\newblock Calculating sums of infinite series.
\newblock {\em The American mathematical monthly}, 99(7):649--655.

\bibitem[Burns and Daniels, 2023]{Burns2023}
Burns, N. and Daniels, M.~J. (2023).
\newblock Truncation approximation for enriched dirichlet process mixture
  models.
\newblock {\em arXiv preprint arXiv:2305.01631}.

\bibitem[Bürkner, 2017]{Burkner2018}
Bürkner, P.-C. (2017).
\newblock {brms}: An {R} package for {Bayesian} multilevel models using {Stan}.
\newblock {\em Journal of Statistical Software}, 80(1):1--28.

\bibitem[Carpenter et~al., 2017]{Carpenter2017}
Carpenter, B., Gelman, A., Hoffman, M.~D., Lee, D., Goodrich, B., Betancourt,
  M., Brubaker, M., Guo, J., Li, P., and Riddell, A. (2017).
\newblock Stan: A probabilistic programming language.
\newblock {\em Journal of statistical software}, 76(1):1--32.

\bibitem[Cilibrasi and Vit{\'a}nyi, 2011]{Cilibrasi2011}
Cilibrasi, R.~L. and Vit{\'a}nyi, P.~M. (2011).
\newblock A fast quartet tree heuristic for hierarchical clustering.
\newblock {\em Pattern recognition}, 44(3):662--677.

\bibitem[Conway and Maxwell, 1962]{Conway1962}
Conway, R.~W. and Maxwell, W.~L. (1962).
\newblock A queuing model with state dependent service rates.
\newblock {\em Journal of Industrial Engineering}, 12(2):132--136.

\bibitem[Dunn and Smyth, 2005]{Dunn2005}
Dunn, P.~K. and Smyth, G.~K. (2005).
\newblock Series evaluation of tweedie exponential dispersion model densities.
\newblock {\em Statistics and Computing}, 15(4):267--280.

\bibitem[Efron, 1986]{Efron1986}
Efron, B. (1986).
\newblock Double exponential families and their use in generalized linear
  regression.
\newblock {\em Journal of the American Statistical Association},
  81(395):709--721.

\bibitem[Ferreira and L{\'o}pez, 2004]{Ferreira2004}
Ferreira, C. and L{\'o}pez, J.~L. (2004).
\newblock Asymptotic expansions of the {H}urwitz--{L}erch zeta function.
\newblock {\em Journal of Mathematical Analysis and Applications},
  298(1):210--224.

\bibitem[Gaunt et~al., 2019]{Gaunt2019}
Gaunt, R.~E., Iyengar, S., Daalhuis, A. B.~O., and Simsek, B. (2019).
\newblock An asymptotic expansion for the normalizing constant of the
  conway--maxwell--poisson distribution.
\newblock {\em Annals of the Institute of Statistical Mathematics},
  71(1):163--180.

\bibitem[Gillespie et~al., 2017]{Gillespie2017}
Gillespie, C.~S. et~al. (2017).
\newblock Estimating the number of casualties in the {A}merican {I}ndian war: a
  bayesian analysis using the power law distribution.
\newblock {\em The Annals of Applied Statistics}, 11(4):2357--2374.

\bibitem[Goldberg, 1991]{Goldberg1991}
Goldberg, D. (1991).
\newblock What every computer scientist should know about floating-point
  arithmetic.
\newblock {\em ACM computing surveys ({CSUR})}, 23(1):5--48.

\bibitem[Griffin, 2016]{Griffin2016}
Griffin, J.~E. (2016).
\newblock An adaptive truncation method for inference in {B}ayesian
  nonparametric models.
\newblock {\em Statistics and Computing}, 26(1-2):423--441.

\bibitem[Hall, 2004]{Hall2004}
Hall, A.~R. (2004).
\newblock {\em Generalized method of moments}.
\newblock OUP Oxford.

\bibitem[Higham, 2002]{Higham2002}
Higham, N.~J. (2002).
\newblock {\em Accuracy and stability of numerical algorithms}.
\newblock SIAM.

\bibitem[Kahan, 1965]{Kahan1965}
Kahan, W. (1965).
\newblock Pracniques: further remarks on reducing truncation errors.
\newblock {\em Communications of the ACM}, 8(1):40.

\bibitem[Kreminski, 1997]{Kreminski1997}
Kreminski, R. (1997).
\newblock Using {S}impson's rule to approximate sums of infinite series.
\newblock {\em The College Mathematics Journal}, 28(5):368--376.

\bibitem[Liu and Nocedal, 1989]{Liu1989}
Liu, D.~C. and Nocedal, J. (1989).
\newblock On the limited memory bfgs method for large scale optimization.
\newblock {\em Mathematical programming}, 45(1):503--528.

\bibitem[Lotze and Raim, 2023]{Sellers2023}
Lotze, K. S.~T. and Raim, A. (2023).
\newblock {\em COMPoissonReg: Conway-Maxwell Poisson (COM-Poisson) Regression}.
\newblock R package version 0.8.1.

\bibitem[Lyne et~al., 2015]{Lyne2015}
Lyne, A.-M., Girolami, M., Atchad{\'e}, Y., Strathmann, H., and Simpson, D.
  (2015).
\newblock On russian roulette estimates for bayesian inference with
  doubly-intractable likelihoods.
\newblock {\em Statistical science}, 30(4):443--467.

\bibitem[Meehan et~al., 2020]{Meehan2020}
Meehan, T.~D., Michel, N.~L., and Rue, H. (2020).
\newblock Estimating animal abundance with n-mixture models using the r-inla
  package for r.
\newblock {\em Journal of Statistical Software}, 95(2):1–26.

\bibitem[mpmath~development team, 2023]{Johansson2023}
mpmath~development team, T. (2023).
\newblock {\em mpmath: a {P}ython library for arbitrary-precision
  floating-point arithmetic (version 1.3.0)}.
\newblock {\tt https://mpmath.org/}.

\bibitem[Navarro and Fuss, 2009]{Navarro2009}
Navarro, D.~J. and Fuss, I.~G. (2009).
\newblock Fast and accurate calculations for first-passage times in wiener
  diffusion models.
\newblock {\em Journal of mathematical psychology}, 53(4):222--230.

\bibitem[Neal, 2015]{Neal2015}
Neal, R.~M. (2015).
\newblock Fast exact summation using small and large superaccumulators.
\newblock {\em arXiv preprint arXiv:1505.05571}.

\bibitem[Park and Haran, 2020]{Park2020}
Park, J. and Haran, M. (2020).
\newblock A function emulation approach for doubly intractable distributions.
\newblock {\em Journal of Computational and Graphical Statistics},
  29(1):66--77.

\bibitem[Pullin et~al., 2020]{Pullin2020}
Pullin, J., Gurrin, L., and Vukcevic, D. (2020).
\newblock Rater: An r package for fitting statistical models of repeated
  categorical ratings.
\newblock {\em arXiv preprint arXiv:2010.09335}.

\bibitem[{Python Core Team}, 2024]{python312}
{Python Core Team} (2024).
\newblock {\em {Python: A dynamic, open source programming language}}.
\newblock {Python Software Foundation}.
\newblock Python version 3.12.

\bibitem[{R Core Team}, 2022]{R}
{R Core Team} (2022).
\newblock {\em R: A Language and Environment for Statistical Computing}.
\newblock R Foundation for Statistical Computing, Vienna, Austria.

\bibitem[Robert and Roberts, 2021]{Robert2021}
Robert, C.~P. and Roberts, G. (2021).
\newblock {R}ao--{B}lackwellisation in the {M}arkov chain {M}onte {C}arlo era.
\newblock {\em International Statistical Review}.

\bibitem[Roberts and Stramer, 2002]{Roberts2002}
Roberts, G.~O. and Stramer, O. (2002).
\newblock Langevin diffusions and {M}etropolis-{H}astings algorithms.
\newblock {\em Methodology and computing in applied probability},
  4(4):337--357.

\bibitem[Royle, 2004]{Royle2004}
Royle, J.~A. (2004).
\newblock N-mixture models for estimating population size from spatially
  replicated counts.
\newblock {\em Biometrics}, 60(1):108--115.

\bibitem[Rudin, 1964]{Rudin1964}
Rudin, W. (1964).
\newblock {\em Principles of {M}athematical {A}nalysis}, volume~3.
\newblock McGraw-hill New York.

\bibitem[Rump et~al., 2008]{Rump2008}
Rump, S.~M., Ogita, T., and Oishi, S. (2008).
\newblock Accurate floating-point summation part i: Faithful rounding.
\newblock {\em SIAM Journal on Scientific Computing}, 31(1):189--224.

\bibitem[Rump et~al., 2009]{Rump2009}
Rump, S.~M., Ogita, T., and Oishi, S. (2009).
\newblock Accurate floating-point summation part ii: Sign, k-fold faithful and
  rounding to nearest.
\newblock {\em SIAM Journal on Scientific Computing}, 31(2):1269--1302.

\bibitem[Sellers et~al., 2012]{Sellers2012}
Sellers, K.~F., Borle, S., and Shmueli, G. (2012).
\newblock The {COM-P}oisson model for count data: a survey of methods and
  applications.
\newblock {\em Applied Stochastic Models in Business and Industry},
  28(2):104--116.

\bibitem[Shmueli et~al., 2005]{Shmueli2005}
Shmueli, G., Minka, T.~P., Kadane, J.~B., Borle, S., and Boatwright, P. (2005).
\newblock A useful distribution for fitting discrete data: revival of the
  {C}onway--{M}axwell--{P}oisson distribution.
\newblock {\em Journal of the Royal Statistical Society: Series C (Applied
  Statistics)}, 54(1):127--142.

\bibitem[Small, 2010]{Small2010}
Small, C.~G. (2010).
\newblock {\em Expansions and asymptotics for statistics}.
\newblock Chapman and Hall/CRC.

\bibitem[Vehtari et~al., 2021]{Vehtari2021}
Vehtari, A., Gelman, A., Simpson, D., Carpenter, B., and Bürkner, P.-C.
  (2021).
\newblock {Rank-Normalization, Folding, and Localization: An Improved
  $\widehat{R}$ for Assessing Convergence of MCMC (with Discussion)}.
\newblock {\em Bayesian Analysis}, 16(2):667 -- 718.

\bibitem[Wei and Murray, 2017]{Wei2017}
Wei, C. and Murray, I. (2017).
\newblock Markov chain truncation for doubly-intractable inference.
\newblock In {\em Artificial Intelligence and Statistics}, pages 776--784.
  PMLR.

\bibitem[Weniger, 2007]{Weniger2007}
Weniger, E.~J. (2007).
\newblock Asymptotic approximations to truncation errors of series
  representations for special functions.
\newblock In {\em Algorithms for Approximation}, pages 331--348. Springer.

\end{thebibliography}

\appendix

\setcounter{table}{0}
\renewcommand{\thetable}{S\arabic{table}}
\renewcommand{\thefigure}{S\arabic{figure}}

\section{Proofs}
\label{app:sec:proofs}

Proof of Proposition~\ref{prop:infinite_series_bounds}:
\begin{proof}
First define the series $r_n = \frac{a_{n+1}}{a_n}$. Now define the remainder $R_n = S - S_n = \sum_{k = n+1}^\infty a_k$.
Now assume that $r_n$ decreases to $L$.
Then
\begin{equation}
\begin{aligned}
R_n &= a_{n-1} \left( \frac{a_{n+1}}{a_{n-1}} + \frac{a_{n+2}}{a_{n-1}} + \frac{a_{n+3}}{a_{n-1}} + \ldots\right) \\
&= a_{n-1} \left( \frac{a_{n}}{a_{n-1}}\frac{a_{n+1}}{a_{n}} + \frac{a_{n}}{a_{n-1}}\frac{a_{n+1}}{a_{n}}\frac{a_{n+2}}{a_{n+1}} + \frac{a_{n}}{a_{n-1}}\frac{a_{n+1}}{a_{n}}\frac{a_{n+2}}{a_{n+1}}\frac{a_{n+3}}{a_{n+2}} + \ldots\right) \\
&= a_{n-1} \left( r_{n-1}r_{n} + r_{n-1}r_{n}r_{n+1} + r_{n-1}r_{n}r_{n+1}r_{n+2} + \ldots\right) \\
&< a_{n-1} \left( r_{n-1}r_{n-1} + r_{n-1}r_{n-1}r_{n-1} + r_{n-1}r_{n-1}r_{n-1}r_{n-1} + \ldots\right) \\
&= a_{n-1} r_{n-1}^2 \left( 1 + r_{n-1} + r_{n-1}^2 + \ldots\right) \\
&=a_{n}r_{n-1}\sum_{k=0}^\infty r_{n-1}^k = a_{n} \frac{r_{n-1}}{1 - r_{n-1}} \\
&= a_{n} \frac{\frac{a_{n}}{a_{n-1}}}{1 - \frac{a_{n}}{a_{n-1}}} = a_{n} \frac{\frac{a_{n}}{a_{n-1}}}{\frac{a_{n-1} - a_{n}}{a_{n-1}}} \\
&= a_{n} \frac{a_{n}}{a_{n-1} - a_{n}} = a_{n} \left(\frac{1}{1 - \frac{a_{n}}{a_{n-1}}} \right).
\end{aligned}
\end{equation}
\noindent On the other hand, since $r_n > L$ for all $n$,
\begin{equation}
\begin{aligned}
R_n  &= a_{n} \left( \frac{a_{n+1}}{a_{n}} + \frac{a_{n+2}}{a_{n}} + \frac{a_{n+3}}{a_{n}} + \ldots\right) \\
&= a_{n} \left( \frac{a_{n+1}}{a_{n}} + \frac{a_{n+1}}{a_{n}}\frac{a_{n+2}}{a_{n+1}} + \frac{a_{n+1}}{a_{n}}\frac{a_{n+2}}{a_{n+1}}\frac{a_{n+3}}{a_{n+2}} + \ldots\right) \\
&= a_{n} r_n \left(1 + r_{n+1}+r_{n+1}r_{n+2}+r_{n+1}r_{n+2}r_{n+3}+\ldots\right)\\
&< a_{n}L\left(1+L+L^2+L^3+\ldots\right)\\
&=a_{n}L\sum_{k=0}^\infty L^k = a_{n} \frac{L}{1-L}.
\end{aligned}
\end{equation}
\noindent For the case in which $r_n$ increases to $L$, the proof is analogous, with inequality signs reversed.
See also Theorem 3 in~\cite{Braden1992}.
Note that if $r_n$ oscillates around $L$, one can easily compute the bounding pair endpoints and order the values so as to obtain a proper bounding pair with the claimed guarantees.
\end{proof}

Now, the proof of Proposition~\ref{prop:L_less_than_M}:
\begin{proof}
First, assume that $a_n < \varepsilon$. Then, for every $\varepsilon > 0$ and $n > n_0$

\begin{equation}
\begin{aligned}
S - S_n &= a_{n+1} + a_{n+2} + a_{n+3} + \ldots \\
&= a_n \left( \frac{a_{n+1}}{a_n} + \frac{a_{n+2}}{a_n} + \frac{a_{n+3}}{a_n} + \ldots \right) \\
&= a_n \left( \frac{a_{n+1}}{a_n} + \frac{a_{n+1}}{a_n} \frac{a_{n+2}}{a_{n+1}} + \frac{a_{n+1}}{a_n} \frac{a_{n+2}}{a_{n+1}} \frac{a_{n+3}}{a_{n+2}} + \ldots \right) \\
&\leq a_n (M + M^2 + M^3 + \ldots)\\
&< \varepsilon M (1 + M + M^2 + \ldots)\\
&= \varepsilon M \sum_{k=0}^\infty M^k = \varepsilon \left( \frac{M}{1-M} \right).
\end{aligned}
\end{equation}
\end{proof}

Proof of Theorem~\ref{theo:bounding_vs_threshold}:

\begin{proof}
First, given the monotonicity hypothesis, we will find the values of $M$ that give the best results in the Sum-to-threshold approach and finally show that this is below the bounding pairs approach

\begin{lemma}
    Let $\{a_n\}$ positive, decreasing and pass in the ratio test with limit $L < 1$.
    Suppose $r_n = \frac{a_{n+1}}{a_n}$ are non-decreasing to L.
    If $a_n < \varepsilon$ for a given $\varepsilon$, we have that, $S - S_n \leq \varepsilon \frac{L}{1-L} < \varepsilon \frac{M}{1-M},\ \forall M \in (L, 1)$.
    i.e. $M = L$ minimise $S - S_n$ with a fixed $\varepsilon$.
\end{lemma}

\begin{proof}
    Suppose that $\exists M^\star \in (L, 1)$ such that $\varepsilon \frac{M^\star}{1 - M^\star} < \varepsilon \frac{L}{1 - L}$.

    If $M^\star > L$, then $\varepsilon \frac{M^\star}{1 - M^\star} > \varepsilon \frac{L}{1 - L}$;
    
    If $M^\star < L$, then $\exists n_0 \in \mathbb{N}$ such that $\frac{a_{n+1}}{a_n} \geq M^\star$ because $\lim \frac{a_{n+1}}{a_n} = L$.

    Then, non exists $M^\star$, to complete the proof, note that, $\frac{a_{n+1}}{a_n} \leq L$ because $\frac{a_{n+1}}{a_n}$ is non-decreasing.
\end{proof}

\begin{lemma}
    Let $\{a_n\}$ positive, decreasing and pass in the ratio test with limit $L < 1$.
    Suppose $r_n = \frac{a_{n+1}}{a_n}$ non-increasing to L.
    If $a_n < \varepsilon$ for a given $\varepsilon$, we have that, $S - S_n \leq \varepsilon \frac{r_n}{1-r_n} < \varepsilon \frac{M}{1-M},\ \forall M \in (L, 1)$.
    i.e. $M = r_n$ minimise $S - S_n$ with a fixed $\varepsilon$.
\end{lemma}

\begin{proof}
    Suppose that $\exists M^\star \in (L, 1)$ such that $\varepsilon \frac{M^\star}{1 - M^\star} < \varepsilon \frac{r_n}{1 - r_n}$.

    If $M^\star > r_n$, then $\varepsilon \frac{M^\star}{1 - M^\star} > \varepsilon \frac{r_n}{1 - r_n}$;

    If $M^\star < r_n$, then in particular $M^\star < \frac{a_{n+1}}{a_n}$, which breaks one of the hypotheses.

    Then, non exists $M^\star$, to complete the proof, note that, $\frac{a_{k+1}}{a_k} \leq \frac{a_{n+1}}{a_n},\ \forall k \geq n$ because $\frac{a_{n+1}}{a_n}$ is non-increasing.
\end{proof}

Now we will prove that this result isn't better than the Bounding Pairs approach.
Consider first $\{a_n\}$ with the hypothesis of the Bounding Pairs approach.
Let $\{r_n\}$, when $r_n = \frac{a_{n+1}}{a_n}$, then with the above results

$$
S - S_n \leq a_{n+1} \frac{L}{1 - L}
$$

if $\{r_n\}$ is non-decreasing, and

$$
S - S_n \leq a_{n+1} \frac{r_n}{1 - r_n}
$$

if $\{r_n\}$ is non-increasing.
Now the Bounding Pairs approach says that our error is

$$
\frac{a_{n+1}}{2} |(1-L)^{-1} - (1 - r_n)^{-1}|.
$$

Checking for each case, first, for $r_n$ non-decreasing

\begin{align*}
    \frac{a_{n+1}}{2} ((1-L)^{-1} - (1 - r_n)^{-1}) \leq a_{n+1} \frac{L}{1-L} &\Leftrightarrow \frac{1}{1-L} - \frac{1}{1 - r_n} \leq \frac{2L}{1-L} \\
    &\Leftrightarrow \frac{(1 - r_n) - (1 - L)}{(1-L)(1-r_n)} \leq \frac{2L (1-r_n)}{(1-L)(1-r_n)} \\
    &\Leftrightarrow L - r_n \leq 2L (1 - r_n) \\
    &\Leftrightarrow 0 \leq L - 2r_nL - r_n
\end{align*}

But we have that

$$
0 \leq (L - r_n)^2 = L^2 - 2r_n L + r_n^2 \leq L - 2r_n L + r_n.
$$

And the equality is valid only when $r_n = L$.
Now we will prove for $r_n$ non-increasing

\begin{align*}
    \frac{a_{n+1}}{2} ((1 - r_n)^{-1} - (1-L)^{-1})&\Leftrightarrow a_{n+1} \frac{r_n}{1 - r_n} \\
    &\Leftrightarrow (1 - r_n)^{-1} - (1 - L)^{-1} \leq \frac{2 r_n}{1 - r_n} \\
    &\Leftrightarrow \frac{(1 - L) - (1 - r_n)}{(1 - r_n) (1 - L)} \leq \frac{2 r_n (1 - L)}{(1 - r_n) (1 - L)} \\
    &\Leftrightarrow r_n - L \leq 2 r_n (1 - L) \\
    &\Leftrightarrow 0 \leq L - 2 r_n L + r_n.
\end{align*}

Analogous to the case above, we have the result, and the equality is valid when $r_n = L$.
\end{proof}

Proof of Remark~\ref{rmk:factmom}:
\begin{proof}
First, notice that 
\begin{equation*}
    E_{\mathcal{M}}[X^r] = \sum_{j=0}^r \stirling{r}{j}  E_{\mathcal{M}}[(X)_j],
\end{equation*}
where $\stirling{a}{b}$ are Stirling numbers of the second kind.
This means that $E_{\mathcal{M}}[X^r] < \infty \implies E_{\mathcal{M}}[(X)_r] < \infty$ and thus the series under consideration is absolutely convergent; notice the support of $X$.
Now the factorial moment can be computed explicitly:
\begin{equation*}
    E_{\mathcal{M}}[(X)_r]  = \sum_{n=r}^\infty \underbrace{\frac{n!}{(n-r)!} f(n\mid \boldsymbol{\theta})}_{a_n},
\end{equation*}
from which one can conclude that $\lim_{n \to \infty} a_{n+1}/a_n = 1 \cdot L < 1$ -- see Appendix~\ref{app:sec:computing_L}.
It now remains to be shown that $E_{\mathcal{M}}[X^r]$ exists.
This much is clear from the fact that $\lim_{n \to \infty}(n+1)^r/n^r = 1$ and thus $\left(y_n\right)_{n\geq0}$, with  $y_n = n^r f(n \mid \boldsymbol{\theta})$, passes the ratio test.
If $f(x+1\mid \boldsymbol{\theta})/f(x\mid \boldsymbol{\theta})$ is decreasing we have that

\begin{align*}
    \frac{a_{n+1}}{a_n} \geq \frac{a_{n+2}}{a_{n+1}} &\Leftrightarrow \frac{n+1}{n+1-r} \frac{f(n+1\mid \boldsymbol{\theta})}{f(n\mid \boldsymbol{\theta})} \geq \frac{n+2}{n+2-r} \frac{f(n+2\mid \boldsymbol{\theta})}{f(n+1\mid \boldsymbol{\theta})}.
\end{align*}
Note that $\frac{n+1}{n+1-r}$ is decreasing in n, so since $f(n+1\mid \boldsymbol{\theta})/f(n\mid \boldsymbol{\theta})$ is also decreasing it is valid that $\frac{a_{n+1}}{a_n} \geq \frac{a_{n+2}}{a_{n+1}}$, so we can apply Bounding pairs approach.
\end{proof}

Proof of Proposition~\ref{prop:marginal}:
\begin{proof}
It is required that
\begin{equation}
\lim_{n\rightarrow\infty}\frac{\operatorname{Pr}(X = n+1, Y\in E)}{\operatorname{Pr}(X=n, Y\in E)} < 1.
\end{equation}
However, due to the definition of conditional probability $\operatorname{Pr}(X=n, Y\in E) = \operatorname{Pr}(X=n\mid Y\in E)\operatorname{Pr}(Y\in E)$, note that when $\operatorname{Pr}(Y\in E) > 0$,
\begin{equation}
\frac{\operatorname{Pr}(X=n+1, Y\in E)}{\operatorname{Pr}(X=n, Y\in E)} = \frac{\operatorname{Pr}(X=n+1\mid Y\in E)}{\operatorname{Pr}(X=n\mid Y\in E)}.
\end{equation}
The proof then follows by assumption.
\end{proof}

Proof of Remark~\ref{rmk:marg_count_error}:
\begin{proof}
First, notice that if $f_Y(y \mid \theta)$ is a p.m.f. and $\lim_{n \to \infty} f_Y(y + 1 \mid \theta)/f_Y(y \mid \theta) =: L \leq 1$.
Next, consider the size-dependent case:
\begin{align*}
    f_X(x \mid \theta) &= \sum_{n=x}^\infty \binom{n}{x}p^x (1-p)^{n-x} f_Y(y \mid \theta),\\
    &= \left(\frac{p}{1-p}\right)^x \frac{1}{x!} \sum_{n=x}^\infty \underbrace{\frac{n!}{(n-x)!} (1-p)^n f_Y(y \mid \theta)}_{g(n)}.
\end{align*}
From this we can surmise
\begin{align*}
    \lim_{n \to \infty} \frac{g(n+1)}{g(n)} &= \lim_{n \to \infty} \frac{n+1}{n+1-x} (1-p)L,\\
    &= L(1-p) < 1,
\end{align*}
which shows that the series of interest now passes the ratio test.
A very similar calculation shows that the problem in equation~\eqref{eq:p_of_X_sentinel} attains ratio $L(1-\eta)$ and thus also passes the ratio test.
\end{proof}

Now, the short proof of Remark~\ref{rmk:comp_naive}:
\begin{proof}
The ratio of consecutive probabilities for the COMP is 
\begin{equation*}
    r_k := \frac{\operatorname{Pr}(X = k + 1 \mid \lambda, \nu)}{\operatorname{Pr}(X = k \mid \lambda, \nu)} = \frac{\lambda}{(k + 1)^\nu}.
\end{equation*}
Thus, $L = \lim_{n \to \infty} r_k = 0$ and note that $\frac{r_{k+1}}{r_{n}}$ is decreasing because $\frac{r_{k+1}}{r_{n}} > 0$ and converges to 0.
Then we can apply Proposition~\ref{prop:infinite_series_bounds}.
\end{proof}

\section{A generalization for bounding pairs}
\label{app:sec:genera_version_of_bounding_pairs}

A somewhat more general result for Proposition~\ref{prop:infinite_series_bounds} can be identified in Proposition~\ref{prop:limits_of_infinite_series_general_form}.

\begin{proposition}[\textbf{Bounding a convergent infinite series with a index function}]
\label{prop:limits_of_infinite_series_general_form}
Let $S_n := \sum_{k = 0}^n a_k$.
Under the assumptions that $\left(a_n\right)_{n\geq 0}$ is positive, decreasing and passes the ratio test, then for every $0 \leq n < \infty$.
Let $r_n = \frac{a_{n+1}}{a_n} \to L$, if $\exists \psi: \mathbb{N} \to \{0,1\}$ such that $\psi(n) = 1$ if $r_n > L$ and 0 otherwise.
The following holds:

\begin{equation}
\label{eq:sum_approx}
 S_{n} + a_{n+1} \left( \frac{1}{1-r_\text{inf}}\right) < S < S_{n} + a_{n+1} \left( \frac{1}{1-r_\text{sup}}\right),
\end{equation}

where, $r_\text{inf} := \min\limits_{\begin{smallmatrix} \psi(i)=0 & \\ i \in \mathbb{N} \setminus I_n \end{smallmatrix}}\{x_i, L\}$ and $r_\text{sup} := \max\limits_{\begin{smallmatrix} \psi(i)=1 & \\ i \in \mathbb{N} \setminus I_n \end{smallmatrix}}\{x_i, L\}$, being $I_n = \{1, 2, \ldots, n\}$.
\end{proposition}
\begin{proof}
First define the remainder $R_n = S - S_n = \sum_{k = n+1}^\infty a_k = a_{n+1}\left(1 + \frac{a_{n+2}}{a_{n+1}} + \frac{a_{n+3}}{a_{n+1}} + \ldots\right)$.
Further expand 
$$
R_n = a_{n+1}\left(1 + \frac{a_{n+2}}{a_{n+1}} + \frac{a_{n+3}}{a_{n+2}}\frac{a_{n+2}}{a_{n+1}} + \ldots\right) = a_{n+1}\left(1+r_{n+1}+r_{n+1}r_{n+2}+\ldots\right).
$$
Note that $\forall k > n, x_k < x_\text{sup} = \max\limits_{\begin{smallmatrix} \psi(i)=1 & \\ i \in \mathbb{N} \setminus I_n \end{smallmatrix}}\{x_i, L\}$, being $I_n = \{1, 2, \ldots, n\}$. Then
\begin{equation}
\begin{aligned}
R_n &= a_{n+1}\left(1 + r_{n+1}+r_{n+1}r_{n+2}+r_{n+1}r_{n+2}r_{n+3}+\ldots\right)\\
& < a_{n+1}\left(1 + x_\text{sup} + x_\text{sup}^2 + x_\text{sup}^3 + \ldots\right)\\
&= a_{n+1}\sum_{k=0}^\infty x_\text{sup}^k = \frac{a_{n+1}}{1 - x_\text{sup}}.
\end{aligned}
\end{equation}

On the other hand, $\forall k > n, x_k > x_\text{inf} = \min\limits_{\begin{smallmatrix} \psi(i)=0 & \\ i \in \mathbb{N} \setminus I_n \end{smallmatrix}}\{x_i, L\}$,

\begin{equation}
\begin{aligned}
R_n &= a_{n+1}\left(1 + r_{n+1}+r_{n+1}r_{n+2}+r_{n+1}r_{n+2}r_{n+3}+\ldots\right)\\
& > a_{n+1}\left(1 + x_\text{inf} + x_\text{inf}^2 + x_\text{inf}^3 + \ldots\right)\\
&= a_{n+1}\sum_{k=0}^\infty x_\text{inf}^k = \frac{a_{n+1}}{1 - x_\text{inf}}.
\end{aligned}
\end{equation}
\end{proof}

\section{Computational details}
\label{app:sec:comp_dets}

Here we will give details on the computational specs of the analyses presented in the main text.

\textbf{Stan:} all Stan runs used four parallel chains running 5000 iterations for warm-up and 5000 for sampling, resulting in 20,000 draws being saved for computing estimates.
We employed \verb|tree_depth=12| and \verb|adapt_delta=0.90|.
MCSE and ESS estimates were computed using the \verb|monitor()| function available in the package~\textbf{rstan}, according to the procedures laid out in~\cite{Vehtari2021}.
We used~\textbf{cmdstanr} version 0.8.0 with version 2.35 of cmdstan.

\textbf{Machine:} some analyses were performed on a Dell G5 laptop equipped with an Intel Core i7-9750H CPU with 12 cores and 12 MB cache and 16 GB of RAM, running Ubuntu 20.04 and R 4.0.4. 
And some analyses were performed on a Lenovo IdeaPad Gaming 3 laptop equipped with an Intel Core i7-11370H CPU with 8 cores and 16 GB of RAM, running Manjaro Linux and R 4.4.2.

\section{Tests}
\label{app:sec:tests}

In this section we provide a discussion on a range of examples where the correct answer is known in closed-form and thus the accuracy of approximation methods can be assessed more precisely.

\paragraph{Design:} for each example below we approximate the desired sum using Error-bounding, Threshold, fixed cap $K=1000$ and a fixed cap $M = 500, 000$ iterations, which will serve as proxy of the best one can realistically do on a modern computer with double precision and high precision.
Let $S$ be the sum of interest and $\hat{S}$ be the approximation (truncated) sum. 
We compute the approximation error in real space robustly as
\begin{equation*}
    \hat{\varepsilon} = \exp\left(M + \log(1 + \exp(m-M))\right),
\end{equation*}
where $m = \min \{ \log S, \log \hat{S}\}$ and $ M = \max \{ \log S, \log \hat{S}\}$.
Notice that sometimes the approximation error can be positive for the Error-bounding approach, since it might overshoot the answer.
As long as this overshooting stays below $\varepsilon$, however, the approximation is working as expected.

\paragraph{Evaluation of algorithms:} we evaluate algorithms in two respects: whether the approximation returned the answer within the desired accuracy, $\varepsilon$, which we will call the~\textbf{target} error and, failing that, whether the approximation returned an error smaller than the error achieved by employing $M = 500, 000$ iterations.
This procedure accounts for the fact that sometimes it is not possible to achieve a certain degree of accuracy, even with a large number of iterations due to numerical limitations inherit to floating point arithmetic. 

\subsection{Poisson factorial moments}

As discussed in Section~\ref{sec:moments}, computing factorial moments is an important application of truncation methods.
For the Poisson distribution, the factorial moment has a simple closed-form, namely
\begin{align*}
    E[\left(X\right)_r] &= \sum_{n= 0}^\infty \operatorname{Pr}(X = n \mid \lambda) n\cdot (n-1)\cdots (n-r+1),\\
    &= \sum_{n=r}^\infty \frac{\lambda^n \exp(-\lambda)}{n!} \frac{n!}{(n-r)!},\\
    &= \lambda^r,
\end{align*}
for $r\geq 1$. 
For this example it can be shown that $L = 0$, and thus we expect the Sum-to-threshold algorithm to return the correct answer.
In our tests we used $\lambda = \{0.5, 1, 10, 100\}$, $r = \{ 2, 5, 10\}$ and $\varepsilon = 2.2 \times \{10^{-16}, 10^{-15}, 10^{-12}\}$.
Table~\ref{tab:tests_poisson_fact_mom_64bits} shows the overall results across these parameter combinations and target errors ($\varepsilon$), from which we can see that all approaches reach a high success rate in providing satisfactory approximations.
And in Table~\ref{tab:tests_poisson_fact_mom_high_prec} we can see that the values converge when at a higher precision, which is a solution to this type of problem.

\begin{table}[ht]
\centering
\begin{tabular}{cccc}
  \hline
Method & $|S-\hat{S}| < \varepsilon$ & $\frac{|S-\hat{S}|}{|S - S_M|} \leq 1$ & Either \\ 
  \hline
  Error-bounding pairs & 0.61 & 0.55 & 1.00 \\ 
  Threshold & 0.61 & 0.53 & 0.75 \\ 
  Cap = $1\times 10^3$ & 0.61 & 0.75 & 0.75 \\ 
  Cap = $5\times 10^5$ & 1.00 & 1.00 & 1.00 \\
   \hline
\end{tabular}
\caption{\textbf{Success rates of the various truncation methods for the Poisson factorial moment problem with 64-bits of precision}.
We show for the fraction of runs in which (i) the truncated sum was within $\varepsilon$ of the true answer ($|S-\hat{S}| < \varepsilon$), (ii) the truncation error was the same or smaller than that of the sum with a large number of iterations $\left(\frac{|S-\hat{S}|}{|S - S_M|} \leq 1\right)$ or (iii) either of these conditions were met.
}
\label{tab:tests_poisson_fact_mom_64bits}
\end{table}

\begin{table}[ht]
\centering
\begin{tabular}{cccc}
  \hline
Method & $|S-\hat{S}| < \varepsilon$ & $\frac{|S-\hat{S}|}{|S - S_M|} \leq 1$ & Either \\ 
  \hline
  Error-bounding pairs & 1.00 & 0.00 & 1.00 \\ 
  Threshold & 0.75 & 0.00 & 0.75 \\ 
  Cap = $1\times 10^3$ & 0.75 & 0.75 & 0.75 \\ 
  Cap = $5\times 10^5$ & 1.00 & 1.00 & 1.00 \\
   \hline
\end{tabular}
\caption{\textbf{Success rates of the various truncation methods for the Poisson factorial moment problem with 100 decimal places of precision}.
We show for the fraction of runs in which (i) the truncated sum was within $\varepsilon$ of the true answer ($|S-\hat{S}| < \varepsilon$), (ii) the truncation error was the same or smaller than that of the sum with a large number of iterations $\left(\frac{|S-\hat{S}|}{|S - S_M|} \leq 1\right)$ or (iii) either of these conditions were met.
}
\label{tab:tests_poisson_fact_mom_high_prec}
\end{table}

\subsection{Negative binomial model with size-independent observation error}

We now present a marginalisation problem for which $L$ can take values in $\left(\frac{1}{2}, 1\right)$, which is useful for testing the accuracy of Sum-to-threshold when there are no mathematical guarantees that it will produce a truncation error less than $\varepsilon$.
Consider the following model:
\begin{align*}
    Y &\sim \operatorname{Negative-binomial(\mu, \phi)},\\
    X \mid Y &\sim \operatorname{Binomial}(Y, \eta).
\end{align*}
We would like to compute the marginal probability mass function for the observed counts:
\begin{align*}
    \operatorname{Pr}(X = x \mid \mu, \phi, \eta) &= \sum_{y=x}^\infty \binom{y + \phi -1}{y} \left(\frac{\mu}{\mu + \phi}\right)^y \left(\frac{\phi}{\mu + \phi}\right)^\phi \binom{y}{x}\eta^x (1-\eta)^{y-x},\\
    &= \binom{x + \phi - 1}{x} \left(\frac{\eta\mu}{\eta\mu + \phi} \right)^x\left(\frac{\phi}{\eta\mu + \phi}\right)^\phi,
\end{align*}
i.e., the marginal probability mass function of $X$ is a negative binomial with parameters $\eta\mu$ and $\phi$.
For this series, we can show that
\begin{equation*}
    L = \left(\frac{\mu}{\mu + \phi}\right)(1-\eta),
\end{equation*}
which can lie anywhere in $(0,1)$.

In this set of experiments we thus constructed a full grid of values for $\mu = \{ 1, 10, 100 \}$, $\phi =\{0.1, 0.5, 1, 10 \}$, $\eta = \{0.01, 0.1, 0.5, 0.75 \}$, $x = \{0, 5, 10\}$ and $\varepsilon = 2.2 \times \{10^{-16}, 10^{-15}, 10^{-12}\}$.
We stratify the results by whether a certain parameter combination leads to $L > 1/2$ in order to evaluate the performance of the Sum-to-Threshold and `fixed cap' approaches.
The results presented in Table~\ref{tab:tests_negbinom_64bits} show that when $L > 1/2$, using Sum-to-threshold or a moderate fixed cap ($K = 1000$) can lead to inaccuracies in the approximation.
As expected, Error-bounding pairs has satisfactory guarantees since it is the only one in this example that have proven guarantees, since Sum-to-threshold loses its guarantees due to the large $n_0$.
In Table~\ref{tab:tests_poisson_fact_mom_high_prec} we see the experiment more accurately.
Note that in this last case, even with 100 decimal places of precision, we still have a precision error.
Since the bounding pairs method has a guarantee in this case and did not reach 1.00.

\begin{table}[!ht]
\centering
\begin{tabular}{ccccc}
  \hline
Method & $L > 1/2$ & $|S-\hat{S}| < \varepsilon$ & $\frac{|S-\hat{S}|}{|S - S_M|} \leq 1$ & Either \\ 
  \hline
  Error-bounding pairs & No & 0.86 & 0.03 & 0.87 \\ 
  Threshold & No & 0.97 & 0.14 & 0.98 \\ 
  Cap = $1\times 10^3$  & No & 0.98 & 1.00 & 1.00 \\ 
  Cap = $5\times 10^5$  & No & 0.98 & 1.00 & 1.00 \\ 
  Error-bounding pairs & Yes & 0.93 & 0.20 & 0.97 \\ 
  Threshold & Yes & 0.75 & 0.09 & 0.77 \\ 
  Cap = $1\times 10^3$  & Yes & 0.75 & 0.78 & 0.78 \\ 
  Cap = $5\times 10^5$  & Yes & 0.93 & 1.00 & 1.00 \\ 
   \hline
\end{tabular}
\caption{\textbf{Success rates of the various truncation methods for the Negative binomial with observation error problem with 64-bits of precision}.
We show for the fraction of runs in which (i) the truncated sum was within $\varepsilon$ of the true answer ($|S-\hat{S}| < \varepsilon$), (ii) the truncation error was the same or smaller than that of the sum with a large number of iterations $\left(\frac{|S-\hat{S}|}{|S - S_M|} \leq 1\right)$ or (iii) either of these conditions were met.
In addition, we stratify results by situations in which $L > 1/2$ as this has an impact on the accuracy of certain methods (see Section~\ref{sec:technical} in the main text).
}
\label{tab:tests_negbinom_64bits}
\end{table}

\begin{table}[!ht]
\centering
\begin{tabular}{ccccc}
  \hline
Method & $L > 1/2$ & $|S-\hat{S}| < \varepsilon$ & $\frac{|S-\hat{S}|}{|S - S_M|} \leq 1$ & Either \\ 
  \hline
  Error-bounding pairs & No & 0.91 & 0.00 & 0.91 \\ 
  Threshold & No & 0.99 & 0.00 & 0.99 \\ 
  Cap = $1\times 10^3$  & No & 1.00 & 1.00 & 1.00 \\ 
  Cap = $5\times 10^5$  & No & 1.00 & 1.00 & 1.00 \\ 
  Error-bounding pairs & Yes & 1.00 & 0.00 & 1.00 \\ 
  Threshold & Yes & 0.03 & 0.00 & 0.03 \\ 
  Cap = $1\times 10^3$  & Yes & 0.78 & 0.68 & 0.71 \\ 
  Cap = $5\times 10^5$  & Yes & 1.00 & 1.00 & 1.00 \\ 
   \hline
\end{tabular}
\caption{\textbf{Success rates of the various truncation methods for the Negative binomial with observation error problem with 100 decimal places of precision}.
We show for the fraction of runs in which (i) the truncated sum was within $\varepsilon$ of the true answer ($|S-\hat{S}| < \varepsilon$), (ii) the truncation error was the same or smaller than that of the sum with a large number of iterations $\left(\frac{|S-\hat{S}|}{|S - S_M|} \leq 1\right)$ or (iii) either of these conditions were met.
In addition, we stratify results by situations in which $L > 1/2$ as this has an impact on the accuracy of certain methods (see Section~\ref{sec:technical} in the main text).
}
\label{tab:tests_negbinom_high_prec}
\end{table}

\section{Hessian matrix for the queuing model example}
\label{app:sec:hessian}
In this section we provide the necessary calculations to compute the Hessian matrix for the model presented in~\autoref{sec:mmle_erlang}.
If we let 
$$
l(\mu, \beta) := \log\left( L(\theta \mid \boldsymbol{x}) \right) = \sum_{i=1}^J \log\left(f_X(x_i \mid \mu, \beta)\right),
$$
be the log-likelihood, the Hessian will be 
\begin{equation*}
   \mathcal{H}(\mu, \beta) = \begin{bmatrix}
\frac{\partial^2 l}{\partial \mu^2} &
    \frac{\partial^2 l}{\partial \mu\beta} \\
  \frac{\partial^2 l}{\partial \beta\mu} &
    \frac{\partial^2 l}{\partial \beta^2}
\end{bmatrix}.
\end{equation*}
Let $p_\mu(n) := \operatorname{Pr}(Y = n \mid \mu)$.
Then, recalling that
\begin{equation*}
     f_X(x_i \mid \mu, \beta) = \sum_{n=1}^\infty p_\mu(n) f_{X \mid Y}(x_i \mid Y = n, \beta),
\end{equation*}
we can define $g_i(\mu, \beta) = \log(f_X(x_i \mid \mu, \beta))$ to get
\begin{align*}
    \frac{\partial^2 l}{\partial \mu^2} &= \sum_{i=1}^J  \frac{\partial^2 g_i}{\partial \mu^2} =  \sum_{i=1}^J  \frac{f_X(x_i \mid \mu, \beta)f_X^{(2, 0)}(x_i \mid \mu, \beta)-\left[f_X^{(1, 0)}(x_i \mid \mu, \beta)\right]^2}{\left[f_X(x_i \mid \mu, \beta)\right]^2},\\
    \frac{\partial^2 l}{\partial \mu\beta} &= \sum_{i=1}^K \frac{\partial^2 g_i}{\partial \mu\beta} = \sum_{i=1}^J \frac{f_X^{(1, 1)}(x_i \mid \mu, \beta)}{f_X(x_i \mid \mu, \beta)} ,\\
    \frac{\partial^2 l}{\partial \beta^2} &= \sum_{i=1}^J \frac{\partial^2 g_i}{\partial \beta^2} = \sum_{i=1}^J  \frac{f_X(x_i \mid \mu, \beta)f_X^{(0, 2)}(x_i \mid \mu, \beta)-\left[f_X^{(0, 1)}(x_i \mid \mu, \beta)\right]^2}{\left[f_X(x_i \mid \mu, \beta)\right]^2},
\end{align*}
whence
\begin{align*}
   f_X^{(1, 0)}(x_i \mid \mu, \beta) &:= \sum_{n=1}^\infty p_\mu^\prime(n)f_{X\mid Y}(x_i \mid Y = n, \beta),\\
   f_X^{(2, 0)}(x_i \mid \mu, \beta)&:= \sum_{n=1}^\infty p_\mu^{\prime\prime}(n)f_{X\mid Y}(x_i \mid Y = n, \beta),\\
   f_X^{(0, 1)}(x_i \mid \mu, \beta) &:= \sum_{n=1}^\infty p_\mu(n)f_{X\mid Y}^\prime(x_i \mid Y = n, \beta),\\
   f_X^{(0, 2)}(x_i \mid \mu, \beta) &:= \sum_{n=1}^\infty p_\mu(n)f_{X\mid Y}^{\prime\prime}(x_i \mid Y = n, \beta),\\
   f_X^{(1, 1)}(x_i \mid \mu, \beta) &:= \sum_{n=1}^\infty \left\{ p_\mu^\prime(n)f_{X\mid Y}(x_i \mid Y = n, \beta) + p_\mu(n)f_{X\mid Y}^\prime(x_i \mid Y = n, \beta) \right\},\\
   &= f_X^{(1, 0)}(x_i \mid \mu, \beta) + f_X^{(0, 1)}(x_i \mid \mu, \beta)
\end{align*}

Now, moving on to the actual computations, we have
\begin{align*}
    p_\mu^\prime(n) &= \frac{n \mu^{n-1} \exp(-\mu)}{n!} - \frac{\mu^{n-1} \exp(-\mu)}{n!},\\
    &= p_\mu(n-1) - p_\mu(n)
\end{align*}
which implies that
\begin{equation*}
     f_X^{(1, 0)}(x_i \mid \mu, \beta) = \sum_{n=1}^\infty p_\mu(n-1)f_{X\mid Y}(x_i \mid Y = n, \beta) - f_X(x_i \mid \mu, \beta).
\end{equation*}
For the second derivative w.r.t $\mu$ we have
\begin{align*}
    p_\mu^{\prime\prime}(n) &= \frac{n \mu^{n-1} \exp(-\mu)}{n!} - \frac{\mu^{n-1} \exp(-\mu)}{n!},\\
    &= \frac{p_\mu(n)[\mu - n]^2}{\mu^2} - \frac{p_\mu(n)}{\mu}.
\end{align*}
Thus,
\begin{equation*}
     f_X^{(2, 0)}(x_i \mid \mu, \beta) = \frac{1}{\mu ^2}\sum_{n=1}^\infty [\mu - n]^2 p_\mu(n)f_{X\mid Y}(x_i \mid Y = n, \beta) - \frac{f_X(x_i \mid \mu, \beta)}{\mu}.
\end{equation*}
Now,
\begin{align*}
    f_{X\mid Y}^{\prime}(x_i \mid Y = n, \beta) &=  \frac{n}{\beta} f_{X\mid Y}(x_i \mid Y = n, \beta) -x f_{X\mid Y}(x_i \mid Y = n, \beta),\\
    f_{X\mid Y}^{\prime\prime}(x_i \mid Y = n, \beta) &= \frac{(x_i\beta - n)^2-n}{\beta^2} f_{X\mid Y}(x_i \mid Y = n, \beta),
\end{align*}
hence
\begin{align*}
     f_X^{(0, 1)}(x_i \mid \mu, \beta) & = \frac{1}{\beta}\sum_{n=1}^\infty n p_\mu(n) f_{X\mid Y}(x_i \mid Y = n, \beta)  - x_i f_X(x_i \mid \mu, \beta),\\
     f_X^{(0, 2)}(x_i \mid \mu, \beta) & = \frac{1}{\beta^2}\sum_{n=1}^\infty   [(x_i\beta - n)^2-n] p_\mu(n) f_{X\mid Y}(x_i \mid Y = n, \beta).
\end{align*}

\section{Computing $L$}
\label{app:sec:computing_L}

In many situations it might be hard to compute $L$ in closed-form because the summand is really complicated. 
Here we discuss a few techniques that might be helpful for computing $L$ in order to use the techniques described here.
The first technique one can employ is exploiting the properties of limits to break down the problem into smaller limits that can be computed more easily.
To apply the techniques described here to the double Poisson normalising constant problem in Section~\ref{sec:norm_consts}, we need to compute
\begin{align*}
    L &:= \lim_{n \to \infty} \frac{ \tilde{p}_{\mu, \phi}(n + 1)}{ \tilde{p}_{\mu, \phi}(n)},\\
    &= \lim_{n \to \infty} \frac{\exp(-(n+1))(n+1)^{n+1} \left(\frac{\exp(1)\mu}{(n+1)}\right)^{\phi (n+1)} }{(n+1)!} \cdot \frac{n!}{\exp(-n)(n)^{n} \left(\frac{\exp(1)\mu}{n}\right)^{\phi n}},\\
    &= \lim_{n \to \infty} \left[\left(\frac{n+1}{n}\right)^n\right]^{\phi-1}\left(\frac{\exp(1)\mu}{n + 1}\right)^{\phi}\exp(-1),\\
\end{align*}
from which one might not immediately arrive at $L=0$.
If one however writes $f_1(x) = \exp(-x)$, $f_2(x) = x^x/x!$ and $f_3(x) = (\exp(1)\mu/x)^{(\phi x)}$, one can then arrive at
\begin{align*}
    L &= \lim_{n \to \infty} \frac{ f_1(n + 1)}{ f_1(n)} \frac{ f_2(n + 1)}{ f_2(n)} \frac{ f_3(n + 1)}{ f_3(n)},\\
     &= \lim_{n \to \infty} \frac{ f_1(n + 1)}{ f_1(n)} \cdot \lim_{n \to \infty} \frac{ f_2(n + 1)}{ f_2(n)} \cdot \lim_{n \to \infty} \frac{ f_3(n + 1)}{ f_3(n)},\\
      &= \lim_{n \to \infty}\exp(-1) \cdot  \lim_{n \to \infty} \left(\frac{n+1}{ n}\right)^n \cdot  \lim_{n \to \infty} \left(\frac{1}{n+1}\right)^{\phi} \left(\frac{n}{n+1}\right)^{\phi n} (\exp(1)\mu)^\phi,\\
            &= \lim_{n \to \infty}\exp(-1) \cdot  \lim_{n \to \infty} \left(\frac{n+1}{ n}\right)^n \cdot \lim_{n \to \infty} \left(\frac{1}{n+1}\right)^{\phi} \cdot  \lim_{n \to \infty} \left(\frac{n}{n+1}\right)^{\phi n} \cdot \lim_{n \to \infty}(\exp(1)\mu)^\phi,\\
      &= \exp(-1) \cdot 1 \cdot 0 \cdot 1 \cdot (\exp(1)\mu)^\phi = 0,
\end{align*}
by exploiting the properties of limits of products of functions.
In summary, breaking the target function down into smaller chunks may prove a winning strategy. 
This formulation has the added benefit of laying bare the slow rate of convergence when $\phi$ is small -- made slower when $\mu$ is large. 

\section{Supplementary figures and tables}
\label{app:sec:suppstuff}

\begin{table}[!h]
    \centering
    \begin{tabular}{cccc}
        \hline a & \text{Sum-to-threshold} &  \text{Bounding-pairs} & \text{Percentage of additional terms} \\
        \hline
        2 & 42 & 40 & -4,7\% \\
        1.1 & 291 & 277 & -4,8\% \\
        1.01 & 2577 & 2447 & -5,0\% \\
        1.001 & 23533 & 22249 & -5,4\% \\
        1.0001 & 214007 & 2013361 & -5,9\% \\
        1.00001 & 1930341 & 1805124 & -6,5\% \\
        \hline
    \end{tabular}
    \caption{Asymptotic Domination of the Bounding Pairs Approach over Sum-to-Threshold.}
    \label{tab:stt_vs_bp}
\end{table}


\begin{table}[]
\centering
\begin{tabular}{@{}ccccc@{}}
\toprule
True & Representation & RMSE & \multicolumn{2}{c}{Coverage} \\ \midrule
 & &  fixed/adaptive & Numerical & Analytic \\
\multirow{2}{*}{$\mu$ = 15} & Bessel & 0.02/0.02 & 0.93/0.93 & 0.95/0.95 \\
 & Full & 0.02/0.02 & 0.93/0.93 & 0.95/0.95 \\
\multirow{2}{*}{ $\mu$ = 150} & Bessel & 0.02/0.02 & 0.95/0.95 & 0.96/0.96 \\
 & Full & 0.02/0.02 & 0.95/0.95 & 0.96/0.96 \\
\multirow{2}{*}{$\mu$ = 1500} & Bessel & 0.20/0.39 & 0.00/0.94 & 0.00/0.95 \\
 & Full & 0.20/0.39 & 0.00/0.66 & 0.00/0.95 \\ \bottomrule
\end{tabular}
\caption{\textbf{Root mean squared error and coverage for the estimates of $\beta$ in the Erlang queuing model}.
See Table~\ref{tab:mmle_erlang_mu} for more details.
Note that we have fixed $\beta = 0.1$ in these experiments also.
}
\label{tab:mmle_erlang_beta}
\end{table}

\end{document}